\documentclass[11pt]{article}
\usepackage{amsmath, amsfonts, amssymb, amsthm}
\usepackage{hyperref}
\usepackage{mathtools}
\usepackage{bm}
\usepackage{mathrsfs}
\usepackage{color}
\mathtoolsset{showonlyrefs} 
\usepackage{booktabs}
\usepackage{enumerate}
\usepackage{graphicx}
\usepackage{multirow}
\usepackage[numbers]{natbib}
\usepackage{geometry}
\usepackage{caption}
\graphicspath{{./art/}}

\newcommand{\E}{\mathbb{E}}

\newcommand{\dd}{\mathrm{d}}
\newcommand{\bbE}{\mathbb{E}}
\newcommand{\bbR}{\mathbb{R}}

\newcommand{\cF}{\mathcal{F}}

\newcommand{\cQ}{\mathscr{Q}}
\newcommand{\ind}{\mathsf{1}}

\newcommand{\cB}{\mathcal{B}}
\newcommand{\cE}{\mathcal{E}}

\newcommand{\cK}{\mathcal{K}}

\newcommand{\cT}{\mathcal{T}}
\newcommand{\sA}{\mathscr{A}}

\newcommand{\Var}{\mathrm{Var}}
\newcommand{\m}{\mathfrak{m}}
\newcommand{\fc}{\mathfrak{c}}
\newcommand{\fb}{\mathfrak{b}}

\DeclareFontFamily{U}{mathx}{}
\DeclareFontShape{U}{mathx}{m}{n}{<-> mathx10}{}
\DeclareSymbolFont{mathx}{U}{mathx}{m}{n} 
\DeclareMathAccent{\widecheck}{0}{mathx}{"71}

\newtheorem{lemma}{Lemma}
\newtheorem{corollary}{Corollary}
\newtheorem{proposition}{Proposition}
\newtheorem{theorem}{Theorem}
\theoremstyle{definition}
\newtheorem{definition}{Definition}

\newtheorem{assumption}{Assumption}

\newtheorem{remark}{Remark}
\newtheorem{algorithm}{Algorithm}

\begin{document}
\title{Uniform ergodicity and spectral gap of importance-tempered MCMC with heavy-tailed targets%
\thanks{This manuscript is a minor revision of an earlier version submitted on April 2, 2026.}}

\author{Quan Zhou \\
Department of Statistics, Texas A\&M University \\
College Station, Texas, USA %\\
%\texttt{quan@stat.tamu.edu}
}
\date{}

\maketitle

\begin{abstract} 
We study the importance-tempered random-walk Metropolis (RWM) algorithm on $\bbR^d$, which targets the tempered density $\pi_\beta(x) \propto \pi(x)^\beta$ with $\beta \in (0,1)$ and corrects for the resulting bias by importance weighting. 
The dynamics of this sampler can be described by a continuous-time Markov chain whose embedded chain is the tempered RWM. 
We prove that when the tails of $\pi$ decay at rate $\|x\|^{-\gamma}$, this continuous-time chain is uniformly ergodic if and only if $d/\gamma < \beta < (\gamma - 2)/\gamma$. 
Further, the global spectral gap of this chain can be controlled via the spectral gap of the chain restricted to a ball with radius growing polynomially in $d$. 
Our theory and simulation studies show that, with a suitable choice of $\beta$, importance tempering can substantially reduce sensitivity to initialization and significantly improve efficiency over the standard RWM. 
\end{abstract}

\section{Introduction} \label{sec:intro}

\subsection{Importance-tempered random-walk Metropolis algorithms} \label{sec:intro-rwm}
Let $\Pi$ be a continuous probability distribution on $\bbR^d$ with density function $\pi > 0$ everywhere. 
A standard tool for sampling from $\Pi$ is the random-walk Metropolis (RWM) algorithm: given current state $x$, we propose $x' = x + \sigma Z$, where $\sigma > 0$ and $Z$ has a symmetric density $\kappa$, and then accept or reject $x'$ according to the Metropolis rule. 
Geometric ergodicity and uniform ergodicity are desirable convergence properties for such Markov chains.  The former requires convergence to $\Pi$  at an exponential rate, while the latter further demands that this rate  be bounded uniformly over all initial states. 
The seminal work of~\citet{mengersen1996rates} showed that RWM is never uniformly ergodic on $\bbR^d$. Intuitively, the time needed for the algorithm to enter the high-density region can be made arbitrarily large by starting far in the tails, which prevents uniform convergence over all initial states. 
Moreover,  for RWM with $\int_{\bbR^d} \|z\| \kappa(z) \dd z < \infty$, exponential tail decay of $\Pi$ is necessary for geometric ergodicity, and it is sufficient under additional tail regularity conditions~\citep{roberts1996geometric, jarner2000geometric}. 
%\citet{mengersen1996rates} proved that when $d = 1$, RWM with $\int |z| \kappa(z) < \infty$ is geometrically ergodic if and only if $\Pi$ has exponentially light tails; when $d > 1$, this exponential decay condition remains necessary but is not sufficient for geometric ergodicity~\citep{roberts1996geometric, jarner2000geometric}. 

In this work, we consider a simple approach to achieving uniform ergodicity with heavy-tailed targets. 
Rather than targeting $\Pi$ directly, we run RWM with stationary density $\pi_\beta(x) \propto \pi(x)^\beta$ for some $\beta \in (0, 1)$ and correct for the resulting bias by importance weighting; this idea can be traced back at least to \citet{Clifford1993Discussion}.  
Formally, let $X_1, \dots, X_n$ denote samples generated from this RWM algorithm and define $\Pi(f) \coloneqq \int f \dd \Pi$  for an integrable function $f$. 
We estimate $\Pi(f)$ using the self-normalized importance sampling (SNIS) estimator 
\begin{equation}\label{eq:def-snis-beta}
     \widetilde{\Pi}_{\beta, n}(f)  \coloneqq \frac{ \sum_{i=1}^n f(X_i) w(X_i)  }{ \sum_{i=1}^n w(X_i) }, \text{ where }  w(x) = \frac{ \pi(x) }{ \pi_\beta(x) } \propto  \pi(x)^{1 - \beta}. 
\end{equation} 
The importance weights $w(X_1), \dots, w(X_n)$ only need to be evaluated up to a normalizing constant  and can be reused for estimating $\Pi(f_j)$ for any collection of integrable functions $\{f_j\}$. 
Following \citet{gramacy2010importance} and \citet{zanella2019scalable}, we refer to this procedure as importance-tempered RWM and summarize it in Algorithm~\ref{alg:it-rwm}. 

Such a Markov chain importance sampling scheme can be extended to more general auxiliary stationary distributions. We focus on the tempered density $\pi_\beta$, since this choice requires only a minimal modification of the standard RWM; indeed, in almost every applied setting, the same MCMC algorithm can be used to sample from $\pi$ or $\pi_\beta$ with the same computational cost and implementation effort, except that for tempered targets one needs to keep track of unnormalized importance weights. 
We will show that this nearly zero-cost modification can substantially improve the estimator's precision and robustness  in a broad range of settings.

\begin{algorithm}\label{alg:it-rwm}
Importance-tempered random-walk Metropolis. 
\begin{tabbing}
\qquad \enspace  \textbf{Input:} target density $\pi(x)$, symmetric proposal density $\kappa(z)$, step size $\sigma > 0$,   \\  
\qquad \enspace \qquad inverse temperature $\beta \in (0,1)$, number of iterations $n$, initial state $x_0 \in \bbR^d$.\\ 
\qquad \enspace \textbf{For} $i=1$ \textbf{to} $n$ \textbf{do}\\ 
\qquad \qquad Draw $z_i \sim \kappa(\cdot)$ and set $x^* \leftarrow x_{i-1} + \sigma z_i$.\\[1pt]
\qquad \qquad Compute the  acceptance probability  $p_i \leftarrow 1 \wedge \frac{\pi(x^*)^\beta }{ \pi( x_{i-1} )^\beta}$.\\[1pt]
\qquad \qquad Draw $u_i \sim \mathrm{Unif}(0,1)$.\\
\qquad \qquad \textbf{If} $u_i \le p_i$ \textbf{then} set $x_i \leftarrow x^*$; \textbf{else} set $x_i \leftarrow x_{i-1}$.\\ 
\qquad \qquad Compute the unnormalized importance weight  $\tilde{w}_i \leftarrow  \pi(x_i)^{1-\beta}$.\\
\qquad \enspace \textbf{End for}\\ 
\qquad \enspace \textbf{Output:} samples $(x_1,\ldots,x_n)$ and weights $(\tilde{w}_1,\ldots,\tilde{w}_n)$.
\end{tabbing}
\end{algorithm}

\subsection{Main results} \label{sec:intro-main}

The asymptotic variance of  the SNIS estimator $\widetilde{\Pi}_{\beta, n}(f) $  depends on both the chain's convergence rate and efficiency of importance weighting. 
For theoretical analysis, a convenient and fruitful strategy is to approximate $\widetilde{\Pi}_{\beta, n}(f) $  by the unweighted time average of a continuous-time Markov chain $(Y_t)_{t \geq 0}$,  whose embedded chain is $(X_i)_{i \geq 1}$ and whose holding time at $X_i$ has mean equal to the importance weight $w(X_i)$; see Section~\ref{sec:mcis-intro} for details.   

Our first main result identifies the exact regime in which the continuous-time chain $(Y_t)_{t \geq 0}$ is uniformly ergodic. 
Specifically, when $\Pi$ has density $\pi(x) \propto \|x\|^{-\gamma}$ for all sufficiently large $\|x\|$, we prove that 
%we prove that when the tails of $\Pi$ decay at the polynomial rate $\|x\|^{-\gamma}$, 
$(Y_t)_{t \geq 0}$ is uniformly ergodic if and only if $d/\gamma < \beta < (\gamma - 2) / \gamma$; see Theorem~\ref{th:ergodic}. Clearly, such $\beta$ exists if and only if $\gamma > d + 2$. 
When $\pi(x) \propto \exp( - \|x\|^a )$ for some $a > 0$, $(Y_t)_{t \geq 0}$ is uniformly ergodic for any $\beta \in (0, 1)$; see Corollary~\ref{coro:light}.   
This marks a sharp contrast with standard RWM, which is never uniformly ergodic on $\bbR^d$ and is not even geometrically ergodic for polynomial-tailed targets. 
Our second result concerns the spectral gap of $(Y_t)_{t \geq 0}$, which we denote by $\lambda$. 
When $\Pi$ is the $t$-distribution with $m > 2$ degrees of freedom, we show in Theorem~\ref{th:gap} that $\lambda = \Omega( \sigma^2 \lambda_R / R^2)$, where $\lambda_R$ denotes the spectral gap of the chain $(Y_t)_{t \geq 0}$ restricted to the  ball $\{x \colon \|x\| \leq R\}$, and the radius $R$ grows  polynomially in $d$. This reduces the difficult global spectral-gap problem on $\bbR^d$ to a bounded-domain problem, where more tools are available.  
Again, standard RWM has fundamentally different behavior: it has no positive spectral gap when the target has polynomial tails on $\bbR^d$. 
%The proof of Theorem~\ref{th:gap} is based on a general result for Markov chains that extends the celebrated theorem of~\citet{bakry2008simple} for Langevin diffusions; see Theorem~\ref{th:decomp}. 

For Algorithm~\ref{alg:it-rwm}, the uniform ergodicity of $(Y_t)_{t \geq 0}$ means that the total importance weight accumulated in the tails is uniformly controlled over all initial states, so the performance of the estimator $\widetilde{\Pi}_{\beta, n}(f)$ is largely insensitive to the initialization scheme and length of the burn-in period. Section~\ref{sec:sim-unif} in the supplement presents a numerical study illustrating this phenomenon. 
Nevertheless, the number of iterations needed for Algorithm~\ref{alg:it-rwm} to enter the high-density region may not be bounded over all initial states, and $(Y_t)_{t \geq 0}$ can make arbitrarily many jumps over a fixed time interval. 
We provide another simulation study in Section~\ref{sec:sim-itmh} to illustrate that once the sampler enters the high-density region, the accumulated importance weight there quickly dominates the contribution from the tails and significantly reduces the asymptotic variance relative to the standard RWM run for the same number of iterations.

\subsection{Related works} \label{sec:literature}

Our uniform ergodicity results are closely related to the line of work on sampling algorithms derived from state-dependent Langevin diffusions.  
\citet{stramer1999langevin} and \citet{roberts2002langevin} showed that, even when $\Pi$ has polynomial tails, one may construct uniformly ergodic Langevin-type diffusions by speeding up the process in the tails.   
Specifically, our continuous-time chain $(Y_t)_{t \geq 0}$ makes a proposal leaving $x$ at a rate proportional to $\pi(x)^{\beta - 1}$ (see Section~\ref{sec:mcis-intro}). This tail acceleration is the same as in   the Langevin tempered diffusion (LTD) studied by~\citet{roberts2002langevin}, which has stationary distribution $\Pi$ and diffusion matrix $\pi(x)^{\beta - 1} \bm{I}$  (with notation modified to facilitate comparison). For polynomial-tailed $\Pi$ on $\bbR^d$, \citet{roberts2002langevin} showed that LTD is geometrically ergodic if $\beta$ is chosen appropriately. When $d = 1$, \citet{stramer1999langevin} showed that LTD is uniformly ergodic if $\beta < (\gamma - 2) / \gamma$; the proof is based on verifying a drift condition with Lyapunov function $V(x) = C - (|x| \vee R)^{-\nu}$  for some constants $C, R, \nu > 0$. Our results extend and sharpen these  findings in two directions. First, we show that the same Lyapunov function can be used on $\bbR^d$, and that $\beta < (\gamma - 2) / \gamma$ is both sufficient and necessary for  uniform ergodicity. Second,  sampling algorithms derived from LTD need discretization and Metropolis adjustment. 
In order to avoid low acceptance rates in the tails, the discrete-time Metropolized samplers require carefully designed state-dependent proposals or computationally intensive discretization methods. Otherwise, a naive implementation may fail to be uniformly or geometrically ergodic~\citep{stramer1999langevin2, roberts2002langevin}.  In contrast, we directly establish uniform ergodicity of the Metropolis-adjusted chain  $(Y_t)_{t \geq 0}$. 

Another interesting comparison is with the stereographic projection sampler (SPS) proposed by~\citet{yang2024stereographic}. When the tails of $\Pi$ decay at rate $\|x\|^{-\gamma}$, SPS is uniformly ergodic if and only if $\gamma \geq 2d$, and a variant of SPS based on sub-Cauchy projection is uniformly ergodic if $\gamma \geq d + 1$~\citep{grazzi2026sub}. 
The key innovation of SPS is to use stereographic projection to transform $\bbR^{d}$ to a pre-compact space so that a random walk proposal can move between any two states within finitely many steps. Intuitively, SPS achieves uniform ergodicity by ``compressing the space,'' whereas LTD and our importance-tempered RWM do so by ``compressing the time.'' Several other approaches have also been proposed for improving the convergence rates of RWM with heavy-tailed targets, such as heavy-tailed proposal distributions~\citep{jarner2007convergence} and smooth variable transformations~\citep{johnson2012variable}.

More broadly, our work belongs to a class of methods where, though one could run MCMC directly on the target $\Pi$, one instead samples from a modified stationary distribution and applies importance weighting. The discussion paper of~\citet{Clifford1993Discussion} appeared to be the first to demonstrate potential advantages of  importance-tempered RWM, and~\citet{gramacy2010importance} extended it in the context of simulated tempering. 
Other examples include the dynamic weighting method~\citep{liu2001theory, liang2002dynamically} and informed MCMC samplers~\citep{zanella2019scalable, rosenthal2021jump, zhou2022rapid, li2023importance}. 
These methods have high acceptance probabilities and can even be entirely rejection-free.   They can also be viewed as Markov chain importance sampling schemes where the auxiliary stationary distributions are constructed only implicitly. 
On the theoretical side, \citet{livingstone2025foundations} recently built a general theory for continuous-time Markov chains arising from such informed MCMC methods and proved a uniform ergodicity result  for discrete targets with super-exponential tails. Their continuous-time representation is of the same type as our chain $(Y_t)_{t \geq 0}$, but they consider locally balanced rather than random-walk proposals. 

\section{Analysis of importance-tempered RWM}\label{sec:mc}

\subsection{Theoretical set-up}\label{sec:mcis-intro}  
Let $(X_i)_{i \geq 1}$ be the RWM chain given in Algorithm~\ref{alg:it-rwm} with stationary density $\pi_\beta$. Denote its transition kernel by  $\cT$.   
Assuming that the chain is ergodic, the SNIS estimator  $\widetilde{\Pi}_{\beta, n}(f)$ defined in equation~\eqref{eq:def-snis-beta} is strongly consistent, since both the numerator and denominator (after $n^{-1}$ scaling) are strongly consistent by the law of large numbers for Markov chains. 

To study the asymptotic variance of $\widetilde{\Pi}_{\beta, n}(f)$, we introduce a continuous-time Markov chain (i.e., a pure-jump Markov process) $(Y_t)_{t\geq 0}$ with rate kernel 
\begin{equation}\label{eq:def-jump-kernel}
    \cK(x, B) =  w(x)^{-1} \cT(x, B), \quad \text{ for }   x \in \bbR^d \text{ and Borel } B \subset \bbR^d, 
\end{equation} 
where we recall $w(x) = \pi(x)/\pi_\beta(x)$.
In words, $(Y_t)_{t\geq 0}$ is obtained from $(X_i)_{i\geq 1}$ by inserting exponential holding times: 
at state $X_i$, the chain waits for a holding time $W_i \sim \mathrm{Exp}\{ w(X_i)^{-1} \}$ before moving to $X_{i+1}$. 
Therefore, the mean holding time at $X_i$ is $w(X_i)$, and its stationary average equals one because   $\int w(x) \pi_\beta(x) \dd x = \int \pi(x) \dd x = 1$. 
Assuming $Y_0 = X_1$,  the time average of $f$ over the trajectory of $(Y_t)_{t \geq 0}$ up to the $n$-th jump can be written as 
\begin{equation}\label{eq:def-snis-mc} 
    \frac{1}{ \sum_{i=1}^n W_i } \int_0^{ \sum_{i=1}^n W_i } f(Y_t) \dd t  =  \frac{ \sum_{i=1}^n f(X_i) W_i }{ \sum_{i=1}^n W_i } \eqqcolon \widecheck{\Pi}_{\beta, n}(f). 
\end{equation} 
The only difference between  $\widecheck{\Pi}_{\beta, n}(f)$ and the SNIS estimator $\widetilde{\Pi}_{\beta, n}(f)$ defined in~\eqref{eq:def-snis-beta} is that the importance weights are replaced by unbiased, exponentially distributed estimates. 
A routine comparison argument shows that the asymptotic variance of $\widecheck{\Pi}_{\beta, n}(f)$ is at least as large as that of $\widetilde{\Pi}_{\beta, n}(f)$, and thus the convergence rate of $(Y_t)_{t \geq 0}$ can be used as a proxy for the efficiency of  $\widetilde{\Pi}_{\beta, n}(f)$. 
This approach was introduced by~\citet[Lemma 2]{zanella2019scalable} and also used in~\citet{livingstone2025foundations}.

\subsection{Uniform ergodicity} \label{sec:unif} 

We recall the definition of geometric and uniform ergodicity for the  chain $(Y_t)_{t \geq 0}$. We say it is geometrically ergodic (or exponentially ergodic), if 
$d_{\mathrm{TV}}(   \mathrm{Law}(Y_t \mid Y_0 = x), \, \Pi ) \leq G(x) \theta^t$ for some constant $\theta \in (0, 1)$ and a measurable function 
$G: \bbR^d \rightarrow [0, \infty)$, 
where $d_{\mathrm{TV}}$ is the total variation distance.  If $\sup_{x \in \bbR^d} G(x) < \infty$, we say $(Y_t)_{t \geq 0}$ is uniformly ergodic.  
For our ergodicity analysis of $(Y_t)_{t \geq 0}$, we make the following assumptions on $\pi$ and $\kappa$.  
\begin{assumption} \label{ass:pi}
The target density $\pi$ can be expressed as $\pi(x) = \exp[ \psi(\|x\|) ]$ for $x \in \bbR^d$,  
    where $\psi: [0,\infty) \rightarrow \bbR$ is a  twice continuously differentiable function such that 
    \begin{equation}\label{eq:cond-psi}
     \lim_{u \rightarrow \infty}  u \, \psi'(u) = -\gamma, 
     \text{ and }  \limsup_{u \rightarrow \infty}  \left| u^2 \, \psi''(u)  \right| \leq K \gamma / 2,  
    \end{equation}
    for some constants $\gamma > d$, and $K > 0$.  
\end{assumption}

\begin{assumption}\label{ass:kappa} 
The proposal probability density function $\kappa$ is symmetric (i.e., $\kappa(z) = \kappa(-z)$) and continuous at $0$ with $\kappa(0) > 0$.  
Further, for $Z \sim \kappa$,  $\bbE[ Z Z^\top ] = I_d$ and $\m_3 \coloneqq \bbE [ \|Z\|^3 ] < \infty$. 
\end{assumption}

Assumption~\ref{ass:pi} implies that $\pi$ has polynomial tail decay of asymptotic order $\gamma$. For example, the $d$-dimensional $t$-distribution with $m$ degrees of freedom satisfies~\eqref{eq:cond-psi} with $\gamma = m + d$. This set-up has also been considered in, e.g.,~\citet{jarner2007convergence}. 
By the classical Karamata's theorem~\citep{nair2022fundamentals},  $\pi^\beta$ is integrable if $\beta > d / \gamma$, and $\int \pi(x)^\beta \dd x = \infty$ if $0 < \beta < d/\gamma$. When $\beta = d/\gamma$, condition~\eqref{eq:cond-psi} alone does not determine whether $\pi^\beta$ is integrable.  
Assumption~\ref{ass:kappa} is standard and usually satisfied in practice. 

\begin{theorem}\label{th:ergodic}
Let $(Y_t)_{t \geq 0}$ be the continuous-time Markov chain with rate kernel given by~\eqref{eq:def-jump-kernel}, and suppose that Assumptions~\ref{ass:pi} and~\ref{ass:kappa} hold with $\gamma > d + 2$. 
If $d / \gamma < \beta < (\gamma - 2 ) / \gamma$, then $(Y_t)_{t\geq 0}$ is uniformly ergodic. 
Conversely, if in addition $\pi(x) \geq c_{L} \|x\|^{-\gamma}$ whenever $\|x\| \geq R_L$ for some constants $c_{L}, R_L > 0$, then $(Y_t)_{t\geq 0}$ is uniformly ergodic only if $d / \gamma < \beta < (\gamma - 2 ) / \gamma$. 
\end{theorem}

\begin{proof}
    See Section~\ref{sec:proof-ergodic} in the supplement. 
\end{proof}
 
The sufficiency part of Theorem~\ref{th:ergodic} is proved by establishing a drift condition for $(Y_t)_{t \geq 0}$ with Lyapunov function $V(x) = C - (\|x\| \vee R)^{-\nu}$ for some $C, R, \nu > 0$. See Proposition~\ref{prop:drift} in the supplement. 
The same Lyapunov function was also used by~\citet{stramer1999langevin} to prove uniform ergodicity of a  Langevin-type diffusion on $\bbR$ with dynamics similar  to $(Y_t)_{t \geq 0}$.  Their analysis, however, is significantly simpler due to the one-dimensional setting and the absence of Metropolis adjustment. Our proof requires second-order Taylor expansions together with careful control of the  remainder terms. Thanks to the scaling factor $w(x)^{-1}$ of the kernel $\cK$, we only need to prove that the RWM kernel $\cT$ satisfies 
\begin{equation}\label{eq:drift-T}
(\cT V)(x) - V(x) \leq  - \tilde{\alpha}(x) V(x), \text{ whenever } \|x\| \text{ is sufficiently large,}
\end{equation}
where $\tilde{\alpha}(x) > 0$ may go to zero as $\|x\|\rightarrow \infty$ but at a rate slower than $\pi(x)^{1 - \beta}$.  This is the key difference from most existing drift-and-minorization analysis of symmetric Metropolis--Hastings algorithms~\citep{jarner2000geometric, yang2023complexity, bhattacharya2023explicit}.  
The necessity part of Theorem~\ref{th:ergodic} is proved by establishing another drift condition with test function $V_{\log}(x) = \log(1 + \|x\|^2)$. 

Uniform ergodicity yields some practical benefits for Algorithm~\ref{alg:it-rwm}.
Let $\tau_R$ denote the first time that $(Y_t)_{t \geq 0}$ enters the ball $\cB_R = \{x \colon \|x\| \leq R\}$, and $\E_x$ denote the expectation under the law with $Y_0 = x$ a.s.  
We can use~\eqref{eq:drift-T} to establish a drift condition for $(Y_t)_{t \geq 0}$ (see~\eqref{eq:drift-sA} in  the proof of Theorem~\ref{th:ergodic}), and by~\citet[Thm. 6.1]{down1995exponential}, it yields 
\begin{equation}\label{eq:tau-bound}
    \sup_{x}  \E_x[ e^{\alpha_R \tau_R} ] \leq 2, \text{ and thus } \sup_{x}  \E_x[ \tau_R ] \leq \alpha_R^{-1} \log 2,   
\end{equation}
for some constant  $\alpha_R > 0$. That is, regardless of the initialization, the time it takes $(Y_t)_{t \geq 0}$ to enter $\cB_R$ has a uniform exponential tail bound. 
For the  estimator $\widetilde{\Pi}_{\beta, n}(f)$, this means that the importance weight accumulated in low-density regions is uniformly controlled. 
However, this does not imply that the jump count $n$ (i.e., number of RWM iterations) is controlled. 
With poor initialization, the number of jumps needed to enter $\cB_R$ could be very large. 

A comparison argument shows that $(Y_t)_{t \geq 0}$ is uniformly ergodic when the tails of $\Pi$ decay at rate $\exp( - \|x\|^a )$ for some $a > 0$. In contrast, standard RWM is never uniformly ergodic for these targets and is not geometrically ergodic when $0 < a < 1$. 

\begin{corollary}\label{coro:light}
Let $(Y_t)_{t \geq 0}$ be the continuous-time Markov chain with rate kernel given by~\eqref{eq:def-jump-kernel}.  Suppose that Assumption ~\ref{ass:kappa} holds, and $\pi(x) \propto \exp( - \|x\|^a )$ for some $a > 0$. Then $(Y_t)_{t\geq 0}$ is uniformly ergodic for any $\beta \in (0, 1)$. 
\end{corollary}
\begin{proof}
See Section~\ref{sec:proof-light} in the supplement. 
\end{proof}

\subsection{Spectral gap bounds} \label{sec:gap} 
We now turn to quantitative convergence through spectral gap bounds.  
For Metropolis--Hastings samplers on $\bbR^d$, explicit lower bounds on the global spectral gap are typically difficult to obtain because the state space is unbounded and the tail region can be a bottleneck. Even for RWM with log-concave targets, such bounds have only been obtained recently~\citep{andrieu2024explicit}. 
However, the uniform ergodicity result suggests that the tail region is not a bottleneck for our continuous-time chain $(Y_t)_{t \geq 0}$. 
In this section, we rigorously show that the global spectral gap of $(Y_t)_{t \geq 0}$ is indeed mainly driven by its behavior inside a sufficiently large ball, while the contribution from the tails can be controlled via the same drift condition used to prove uniform ergodicity. 
To formulate this localization result in full generality, we first recall the definitions of the Dirichlet form and spectral gap associated with a general  non-negative symmetric kernel. 

\begin{definition}\label{def:gap}
Let $\cQ$ be a non-negative transition kernel on $\bbR^d$ with $\int \cQ(x, \bbR^d) \Pi(\dd x)< \infty$, and  $\int_A \cQ(x, B) \Pi(\dd x) = \int_B \cQ(x, A) \Pi(\dd x)$ for any Borel $A, B \subset \bbR^d$. The Dirichlet form and spectral gap of $\cQ$ are defined by 
\begin{align}
    \cE_{\cQ}(f, f)  &= \frac{1}{2}\int_{\bbR^d \times \bbR^d}   [f(y) - f(x)]^2 \cQ(x, \dd y) \Pi(\dd x),  \\ 
    \mathrm{Gap}(\cQ) &= \inf_{f \in L^2(\pi): \Var_\Pi(f) > 0 } \frac{ \cE_{\cQ}(f, f)}{\Var_\Pi(f)}, 
\end{align}
where $\Var_\Pi(f) = \int_{\bbR^d} [f(x) - \Pi(f)]^2 \Pi(\dd x)$.  
\end{definition}

The kernel $\cQ$ in Definition~\ref{def:gap} can be either the rate kernel of a continuous-time Markov chain or the transition kernel of a discrete-time chain. The symmetric condition implies that the resulting Markov chain is reversible with respect to $\Pi$. 
We now prove a general theorem that uses the local spectral gap to control the global one; it can be viewed as an extension of the well-known result of~\citet[Thm. 1.4]{bakry2008simple} for the Langevin diffusion to Markov chains. During the revision of this manuscript, a similar result was independently obtained by~\citet[Lemma 4.4]{son2026rapid}.

\begin{theorem}\label{th:decomp}
Let $\cQ$ be a kernel as given in Definition~\ref{def:gap}.   
For $R > 0$, define 
\begin{equation}
\cB_R=\{x\in\bbR^d:\|x\|\le R\}, \quad p_R=\Pi(\cB_R), \quad \Pi_R(A)=\frac{\Pi(A\cap \cB_R)}{p_R}. 
\end{equation}
Let $\cQ_R$ denote the restriction of $\cQ$ to $\cB_R$, which is obtained by rejecting jumps outside of $\cB_R$:   
\begin{equation}\label{eq:def-QR}
    \cQ_R(x, A) =  \cQ(x, A) + \cQ(x, \cB_R^c) \ind_A(x), \text{ for } x \in \cB_R, \text{ and Borel } A \subset \cB_R. 
\end{equation} 
Suppose that for some function $V \colon \bbR^d \rightarrow (0, \infty)$ and constants $\alpha_R, b_R > 0$, we have 
\begin{equation}\label{eq:def-alpha-R}
     \int_{\bbR^d} \left[ \frac{ V(y)}{V(x)} - 1 \right] \cQ(x, \dd y)   \leq - \alpha_R \ind_{\cB_R^c}(x) + b_R \ind_{\cB_R}(x).    
\end{equation}
If $2 b_R(1 - p_R) \leq \alpha_R p_R  $, then 
\begin{align}
     \mathrm{Gap}(\cQ) &\geq \lambda_R \left[ 1 + \frac{2 (\lambda_R + b_R)}{  \alpha_R p_R} \right]^{-1}, 
\end{align}
where $\lambda_R = \mathrm{Gap}(\cQ_R)$ is the spectral gap of $\cQ_R$  with respect to the stationary distribution $\Pi_R$. 
\end{theorem}

\begin{proof}
    See Section~\ref{sec:proof-decomp} in the supplement. 
\end{proof}

We apply Theorem~\ref{th:decomp} to $(Y_t)_{t \ge 0}$.  To keep the presentation concrete, we focus on multivariate $t$ targets, for which the constants $p_R, \alpha_R, b_R$ can be explicitly bounded. 

\begin{theorem}\label{th:gap}
Consider the continuous-time Markov chain $(Y_t)_{t \geq 0}$ with rate kernel $\cK$ given by~\eqref{eq:def-jump-kernel}. 
Let $\Pi$ be the multivariate $t$-distribution with $m$ degrees of freedom and density 
\begin{equation}\label{eq:def-t}
\pi(x) = \frac{\Gamma[ (m+d)/2 ]}{\Gamma( m/2)   (m\pi)^{d/2} } \left( 1 + \frac{\|x\|^2}{m} \right)^{-(m + d) / 2}. 
\end{equation}
Assume $m > 2 + 2 \epsilon$ for some $\epsilon \in (0, 1/2)$,   set $\gamma = m+d$ and choose  $\beta = (\gamma - 2 - \epsilon) / \gamma$.  
Assume %$\sigma \leq \gamma^2$ and the proposal density 
$\kappa$ satisfies Assumption~\ref{ass:kappa}. Then, %with $\int \|z\|^3 \kappa(z) \dd z \leq M d^{3/2}$ for some universal constant $M > 0$. Then, 
\begin{align}
    \mathrm{Gap}(\cK) \geq  \frac{\sigma^2 \epsilon}{252 R^2} \mathrm{Gap}(\cK_R), 
 \; \text{ for any }  R \geq \max\left\{ 2\gamma, \; 26 \sigma \gamma^2  \m_3, \;  \left[ \frac{126  (e\gamma)^{m/2} }{\sigma^2 \epsilon} \right]^{\frac{1}{m - 2}}  \right\}.      
\end{align} 
%where $C \in (0, \infty)$ is a sufficiently large universal constant. 
\end{theorem}
\begin{proof}
    See Section~\ref{sec:proof-th-gap} in the supplement. 
\end{proof}
\begin{remark}
Theorem~\ref{th:gap} provides a non-asymptotic bound on the spectral gap of  $(Y_t)_{t \geq 0}$, where $\sigma, m, \epsilon$ are not assumed to be fixed constants. 
In most applications, we can assume $\sigma^2 = d^{-\eta}$ for some $\eta \in \bbR$  and $\m_3 = O(d^{3/2})$.
Then, treating $m$ and $\epsilon$ as fixed,  we have $\gamma = O(d)$  and 
\begin{equation}
     R = O(d^{a}), \text{ where } a = \max\left\{1, \; \frac{m + 2 \eta}{2(m-2)}, \;    \frac{7 - \eta}{2} \right\}. 
\end{equation} 
So we have effectively localized a global spectral-gap problem on $\bbR^d$ to a ball with radius polynomial in $d$. Particularly, if $\mathrm{Gap}(\cK_R)$ scales polynomially in $d$, so does $\mathrm{Gap}(\cK)$. 
For comparison, standard RWM is not geometrically ergodic with polynomial-tailed targets and thus does not admit a positive spectral gap. 
An explicit lower bound on $\mathrm{Gap}(\cK_R)$ is left for future work. 
\end{remark}

\section{Concluding remarks}\label{sec:conclude} 
Our theory characterizes the convergence of the chain $(Y_t)_{t \geq 0}$ in continuous time $t$, rather than the performance of Algorithm~\ref{alg:it-rwm} as a function of the iteration count $n$. Therefore, in Section~\ref{sec:sim} of the supplement, we present simulation studies for both $(Y_t)_{t \geq 0}$ and Algorithm~\ref{alg:it-rwm} under different initialization schemes. 
We find that, with a suitable choice of $\beta$, importance-tempered RWM is indeed substantially more efficient and more robust to initialization than standard RWM. Though the tempered chain tends to spend many more iterations in the tails especially under poor initialization, the resulting SNIS estimator can have much smaller asymptotic variance regardless of initialization, illustrating that uniform ergodicity effectively controls the tail contribution. 
We also find that it is generally preferable to choose $\beta$ slightly below $(\gamma - 2)/\gamma$, the theoretical threshold for uniform ergodicity of $(Y_t)_{t \geq 0}$. 

Although this paper focuses on RWM, the same importance-tempering idea is much more general. One can run essentially any MCMC sampler against the tempered target and then perform importance weighting, typically at almost no additional implementation cost.

\section*{Acknowledgment}
We are very grateful to the two anonymous reviewers for their thoughtful and constructive comments. 
This work was supported by NSF through grants DMS-2245591 and DMS-2311307.  

\section*{Supplementary material} \label{SM} 
The Supplementary Material includes proofs and simulation studies.

\bibliographystyle{plainnat}
\bibliography{ref}

@inproceedings{zhou2022rapid,
  title={Rapid convergence of informed importance tempering},
  author={Zhou, Quan and Smith, Aaron},
  booktitle={International Conference on Artificial Intelligence and Statistics},
  pages={10939--10965},
  year={2022},
  organization={PMLR}
}

@article{li2023importance,
  title={Importance is Important: Generalized {M}arkov Chain Importance Sampling Methods},
  author={Li, Guanxun and Smith, Aaron and Zhou, Quan},
  journal={arXiv preprint arXiv:2304.06251},
  year={2023}
}

@article{bhattacharya2023explicit,
  title={Explicit constraints on the geometric rate of convergence of random walk {M}etropolis-{H}astings},
  author={Bhattacharya, Riddhiman and Jones, Galin L},
  journal={Bernoulli},
  volume={31},
  number={3},
  pages={2042--2076},
  year={2025},
  publisher={Bernoulli Society for Mathematical Statistics and Probability}
}

@article{roberts1996geometric,
  title={Geometric convergence and central limit theorems for multidimensional {H}astings and {M}etropolis algorithms},
  author={Roberts, Gareth O and Tweedie, Richard L},
  journal={Biometrika},
  volume={83},
  number={1},
  pages={95--110},
  year={1996},
  publisher={Oxford University Press}
}

@article{jarner2000geometric,
  title={Geometric ergodicity of {M}etropolis algorithms},
  author={Jarner, S{\o}ren Fiig and Hansen, Ernst},
  journal={Stochastic Processes and their Applications},
  volume={85},
  number={2},
  pages={341--361},
  year={2000},
  publisher={Elsevier}
}

@article{roberts2002langevin,
  title={Langevin diffusions and {M}etropolis-{H}astings algorithms},
  author={Roberts, Gareth O and Stramer, Osnat},
  journal={Methodology and Computing in Applied Probability},
  volume={4},
  pages={337--357},
  year={2002},
  publisher={Springer}
}

@article{zanella2019scalable,
  title={Scalable importance tempering and {B}ayesian variable selection},
  author={Zanella, Giacomo and Roberts, Gareth},
  journal={Journal of the Royal Statistical Society Series B: Statistical Methodology},
  volume={81},
  number={3},
  pages={489--517},
  year={2019},
  publisher={Oxford University Press}
}

@book{EthierKurtz1986,
  author    = {Stewart N. Ethier and Thomas G. Kurtz},
  title     = {Markov processes: characterization and convergence},
  series    = {Wiley Series in Probability and Mathematical Statistics},
  publisher = {John Wiley \& Sons, Inc.},
  address   = {New York},
  year      = {1986}, 
  doi       = {10.1002/9780470316658} 
}

@article{livingstone2025foundations,
  title={Foundations of locally-balanced {M}arkov processes},
  author={Livingstone, Samuel and Vasdekis, Giorgos and Zanella, Giacomo},
  journal={arXiv preprint arXiv:2504.13322},
  year={2025}
}

@article{Clifford1993Discussion,
  author       = {Jennison, C.},
  title        = {Discussion on the Meeting on the {Gibbs} Sampler and Other {Markov} Chain {Monte} {Carlo} Methods},
  journal      = {Journal of the Royal Statistical Society: Series B (Methodological)},
  volume       = {55},
  number       = {1},
  pages        = {53--102},
  year         = {1993},
  publisher    = {Wiley} 
}

@article{yang2023complexity,
  title={Complexity results for {MCMC} derived from quantitative bounds},
  author={Yang, Jun and Rosenthal, Jeffrey S},
  journal={The Annals of Applied Probability},
  volume={33},
  number={2},
  pages={1459--1500},
  year={2023},
  publisher={Institute of Mathematical Statistics}
}

@article{down1995exponential,
  title={Exponential and uniform ergodicity of {M}arkov processes},
  author={Down, Douglas and Meyn, Sean P and Tweedie, Richard L},
  journal={The Annals of Probability},
  volume={23},
  number={4},
  pages={1671--1691},
  year={1995},
  publisher={Institute of Mathematical Statistics}
}

@article{gramacy2010importance,
  title={Importance tempering},
  author={Gramacy, Robert and Samworth, Richard and King, Ruth},
  journal={Statistics and Computing},
  volume={20},
  pages={1--7},
  year={2010},
  publisher={Springer}
}

@article{yang2024stereographic,
  title={Stereographic {M}arkov chain {M}onte {C}arlo},
  author={Yang, Jun and {\L}atuszy{\'n}ski, Krzysztof and Roberts, Gareth O},
  journal={The Annals of Statistics},
  volume={52},
  number={6},
  pages={2692--2713},
  year={2024},
  publisher={Institute of Mathematical Statistics}
}

@article{jarner2007convergence,
  title={Convergence of heavy-tailed {M}onte {C}arlo {M}arkov chain algorithms},
  author={Jarner, S{\o}ren F and Roberts, Gareth O},
  journal={Scandinavian Journal of Statistics},
  volume={34},
  number={4},
  pages={781--815},
  year={2007},
  publisher={Wiley Online Library}
}

@article{mengersen1996rates,
  title={Rates of convergence of the {H}astings and {M}etropolis algorithms},
  author={Mengersen, Kerrie L and Tweedie, Richard L},
  journal={The Annals of Statistics},
  volume={24},
  number={1},
  pages={101--121},
  year={1996},
  publisher={Institute of Mathematical Statistics}
}

@article{liu2001theory,
  title={A theory for dynamic weighting in {M}onte {C}arlo computation},
  author={Liu, Jun S and Liang, Faming and Wong, Wing Hung},
  journal={Journal of the American Statistical Association},
  volume={96},
  number={454},
  pages={561--573},
  year={2001},
  publisher={Taylor \& Francis}
}

@article{liang2002dynamically,
  title={Dynamically weighted importance sampling in {M}onte {C}arlo computation},
  author={Liang, Faming},
  journal={Journal of the American Statistical Association},
  volume={97},
  number={459},
  pages={807--821},
  year={2002},
  publisher={Taylor \& Francis}
}

@article{rosenthal2021jump,
  title={Jump {M}arkov chains and rejection-free {M}etropolis algorithms},
  author={Rosenthal, Jeffrey S and Dote, Aki and Dabiri, Keivan and Tamura, Hirotaka and Chen, Sigeng and Sheikholeslami, Ali},
  journal={Computational Statistics},
  pages={2789--2811},
  volume={36}, 
  year={2021},
  publisher={Springer}
}

@book{meyn2012markov,
  title={Markov chains and stochastic stability},
  author={Meyn, Sean P and Tweedie, Richard L},
  year={2012},
  publisher={Springer Science \& Business Media}
}

@article{grazzi2026sub,
  title={Sub-{C}auchy Sampling: Escaping the Dark Side of the Moon},
  author={Grazzi, Sebastiano and Liu, Sifan and Roberts, Gareth O and Yang, Jun},
  journal={arXiv preprint arXiv:2601.11066},
  year={2026}
}

@article{johnson2012variable,
  title={Variable transformation to obtain geometric ergodicity in the random-walk {M}etropolis algorithm},
  author={Johnson, Leif T and Geyer, Charles J},
  journal={The Annals of Statistics},
  pages={3050--3076},
  year={2012},
  volume={40}, 
  publisher={JSTOR}
}

@article{andrieu2024explicit,
  title={Explicit convergence bounds for {M}etropolis {M}arkov chains: Isoperimetry, spectral gaps and profiles},
  author={Andrieu, Christophe and Lee, Anthony and Power, Sam and Wang, Andi Q},
  journal={The Annals of Applied Probability},
  volume={34},
  number={4},
  pages={4022--4071},
  year={2024},
  publisher={Institute of Mathematical Statistics}
}

@book{nair2022fundamentals,
  title={The fundamentals of heavy tails: Properties, emergence, and estimation},
  author={Nair, Jayakrishnan and Wierman, Adam and Zwart, Bert},
  volume={53},
  year={2022},
  publisher={Cambridge University Press}
}

@article{stramer1999langevin,
  title={{L}angevin-type models {I}: Diffusions with given stationary distributions and their discretizations},
  author={Stramer, O and Tweedie, RL},
  journal={Methodology and Computing in Applied Probability},
  volume={1},
  number={3},
  pages={283--306},
  year={1999},
  publisher={Springer}
}

@article{bakry2008simple,
  title={A simple proof of the {P}oincar{\'e} inequality for a large class of probability measures.},
  author={Bakry, Dominique and Barthe, Franck and Cattiaux, Patrick and Guillin, Arnaud},
  journal={Electronic Communications in Probability},
  volume={13},
  pages={60--66},
  year={2008},
  publisher={University of Washington}
}

@article{stramer1999langevin2,
  title={Langevin-type models {II}: Self-targeting candidates for {MCMC} algorithms},
  author={Stramer, O and Tweedie, RL},
  journal={Methodology and Computing in Applied Probability},
  volume={1},
  number={3},
  pages={307--328},
  year={1999},
  publisher={Springer}
}

@article{batir2008inequalities,
  title={Inequalities for the gamma function},
  author={Batir, Necdet},
  journal={Archiv der Mathematik},
  volume={91},
  number={6},
  pages={554--563},
  year={2008},
  publisher={Birkh  user-Verlag, Basel}
}

@article{son2026rapid,
  title={Rapid convergence of tempering chains to multimodal {G}ibbs measures},
  author={Son, Seungjae},
  journal={arXiv preprint arXiv:2604.04823},
  year={2026}
}

\clearpage 
\newpage
 
\setcounter{page}{1} % reset the page number to 1
\pagenumbering{arabic} % (re)select arabic numerals; this also resets to 1
\renewcommand{\thepage}{S\arabic{page}}
\renewcommand{\theequation}{A\arabic{equation}}  

\begin{center}
    {\Large Supplement to ``Uniform ergodicity and spectral gap of importance-tempered MCMC with heavy-tailed targets''} \\ \medskip 
    {\large By Quan Zhou} 
\end{center}

\medskip 
\appendix 

\section{Proofs} \label{sec:proofs}

\subsection{Drift conditions for the RWM chain} \label{sec:proof-drift}

Recall that $\cT$ is the transition kernel of the RWM chain constructed in Algorithm~\ref{alg:it-rwm}, which has stationary density $\pi_\beta$  and proposes $x'$  using $x' = x + \sigma Z$ with $Z \sim \kappa$. Denote the acceptance probability by 
\begin{equation}\label{eq:def-pz}
    p(x, z) = \min \left\{ 1, \; \frac{\pi(x + \sigma z)^\beta}{ \pi(x)^\beta}\right\}. 
\end{equation} 
We first prove a general result about the Taylor approximation to the expected drift after one RWM update, and then we use it to establish two drift conditions for $\cT$.

\begin{proposition}\label{prop:taylor}
Assume that $\pi_\beta(x) \propto \exp[ \beta \psi(\|x\|) ]$ with $\beta > 0$. Suppose that the following conditions hold for some function $g \colon [0, \infty) \rightarrow \bbR$ and constants $R > 0$,
$\omega \geq  0, \fc_j, \fb_j > 0$: 
\begin{align}
    &\left| g^{(j)}(u) \right| \leq \fc_j u^{- \omega - j},  \quad \text{ for } u > (R/2)^2    \text{ and } j = 1, 2, 3,  \label{cond:f} \\ 
    &| u \,  \psi'(u) | \leq \fb_1, \quad 
   | u^2 \,  \psi''(u) | \leq \fb_2, \quad  
   \text{ for }  u > R /2.  \label{cond:psi}  
\end{align} 
Further, suppose that $Z \sim \kappa$ satisfies
\begin{equation}\label{cond:Z}
    Z \overset{\mathrm{d}}{=} - Z, \qquad \E[ Z Z^\top] = I_d, \qquad \m_3 \coloneqq \E [ \|Z\|^3 ] < \infty. 
\end{equation}
Define $G(x) = g(\|x\|^2)$ and assume $(\cT G)(x)$ is finite for every $x$. Then, for any $x$ with $r \coloneqq \|x\| > R$,
\begin{equation}
    (\cT G)(x) - G(x) = \sigma^2  \left\{  \left[ d +   \beta r  \psi'(r) \right]  g'(r^2)  + 2  r^2  g''( r^2 )  \right\} + \mathrm{Tail}(x) + \mathrm{Rem}(x),   
\end{equation}
where the tail and remainder terms satisfy 
\begin{align}
     \mathrm{Tail}(x) &=  \E \left\{ \left[ G(x + \sigma Z) - G(x) \right] p(x, Z) \ind_{A(x)^c}(Z) \right\}, \qquad A(x) = \left\{ z: \sigma \|z\| < \|x\|/2  \right\}, \\
    \left|\mathrm{Rem} (x)\right| &\leq (C_0 \vee C_1)  \frac{\sigma^3 \m_3 }{ \|x\|^{2 \omega + 3 } }, \qquad C_0 =  2 ( \fc_1  + 2 \fc_2 + \beta \fb_1 \fc_1),  \\
    C_1 &=   4^{\omega + 2}\left( \fc_2 + \frac{2}{3} \fc_3 \right)   + 2 \beta \fb_1 ( 3 \fc_1  + 2 \fc_2)    +  \beta^2 \fb_1^2 \fc_1 + 4 \beta \fb_2\fc_1. \label{eq:C1}
\end{align}
If $\fc_1 \leq \fc_2$, then $C_1 > C_0$. 
\end{proposition}

\begin{proof}
Fix an arbitrary $x$ such that $\|x\| > R$. Since $\cT$ is the transition kernel of the RWM chain with proposal density $\kappa$ and  $\kappa$ is symmetric, we have 
\begin{align*}
     (\cT G) (x)  - G(x)  &= \int_{\bbR^d} \kappa(z) [G(x + \sigma z) - G(x)]  p(x, z) \dd z  \\
     &= \int_{\bbR^d} \kappa(z) [G(x - \sigma z) - G(x)] p(x, -z) \dd z. 
\end{align*}
Hence, we can equivalently express $(\cT G)(x) - G(x)$ by 
\begin{align}
    (\cT G) (x)  - G(x) &= \frac{1}{2} \E_{ Z \sim \kappa}[ S( Z) ], \label{eq:TV-S} \\  
     S( z)  &= [G(x + \sigma z)-G(x)] p(x, z) + [G(x - \sigma z)-G(x)] p(x, -z) \label{eq:def-Sx}. 
\end{align} 

By the second-order Taylor expansion of $G(x + \sigma z) - G(x)$ and first-order approximation of $p(x, z)$, 
\begin{align}
    G(x + \sigma z) - G(x) &=  F_1(z)  + F_2(z) + \mathcal{R}_G(z),   \label{eq:taylor-G} \\ 
     p(x, z)  &= 1 +  \left[ 0 \wedge D(z) \right] + \mathcal{R}_\pi(z),  
\label{eq:taylor-pi}  
\end{align}
where, as we show in Lemma~\ref{lm:taylor} below,  
\begin{align}
    F_1(z) &= 2 \sigma (x^\top z) g'(\|x\|^2),  \qquad  
    F_2(z) = \sigma^2  \left[ \|z\|^2 g'(\|x\|^2) + 2  (x^\top z)^2 g''(\|x\|^2) \right], \label{eq:taylor-F1} \\ 
    D(z) &= \frac{ \beta \sigma (x^\top z) \psi'(\|x\|)}{\|x\|}.   
    \label{eq:taylor-D} 
\end{align} 
By~\eqref{eq:taylor-G} and~\eqref{eq:taylor-pi}, 
\begin{align}
  &\quad  [G(x + \sigma z)-G(x)] p(x, z) - \left\{ F_1(z) + F_2(z) + F_1(z) [ 0 \wedge D(z)]  \right\} \\
& = \mathcal{R}_G(z) p(x, z) +  F_2(z) [p(x, z) - 1] +   F_1(z)   \mathcal{R}_\pi(z) \label{eq:def-total-remainder} \\
&\eqqcolon \mathcal{R}(z). 
\end{align}
Replacing $z$ with $-z$ and using $F_1(z) = - F_1(-z), F_2(z) = F_2(-z), D(z) = -D(-z)$, we get from~\eqref{eq:def-Sx}  
\begin{align}
    S( z)  
    &= 2 F_2(z) + F_1(z) \left\{ [ 0 \wedge D(z)] - [ 0 \wedge -D(z)] \right\} + \mathcal{R}(z) + \mathcal{R}(-z)  \\
    &= 2 F_2(z) + F_1(z) D(z) + \mathcal{R}(z) + \mathcal{R}(-z). 
\end{align}
Define $\tilde{S}(z) = 2 F_2(z) + F_1(z) D(z)$ and $\bar{\mathcal{R}}(z) = \mathcal{R}(z) + \mathcal{R}(-z)$ so that $S(z) = \tilde{S}(z) + \bar{\mathcal{R}}(z)$. 
Since $x$ is fixed, we simply write $A \coloneqq A(x) = \left\{ z: \sigma \|z\| < \|x\|/2  \right\}$. 
By~\eqref{eq:TV-S},
\begin{align}
   2 [ (\cT G) (x)  - G(x) ] &=    \E \left[ \tilde{S}(Z) +  \bar{\mathcal{R}}(Z) \right] \\
   &=  \E \left[ \tilde{S}(Z) \right] +  \E \left[ S(Z)  \ind_{A^c}(Z)\right]  - \E \left[ \tilde{S}(Z)  \ind_{A^c}(Z)\right]+  \E \left[ \bar{\mathcal{R}}(Z)  \ind_A(Z)\right].  \label{eq:decomp-four-terms}
\end{align}
Here, $\E \left[ \tilde{S}(Z) / 2\right]$ is the Taylor approximation to  $(\cT G) (x)  - G(x)$, and  using~\eqref{cond:Z}, we get 
\begin{align}
   \frac{1}{2}\E \left[ \tilde{S}(Z) \right] &=   \sigma^2 \E\left[      \|Z\|^2 g'(\|x\|^2) + 2  (x^\top Z)^2 g''(\|x\|^2)  +    \beta (x^\top Z)^2 g'(\|x\|^2)  \frac{  \psi'(\|x\|)}{\|x\|} 
   \right] \\
   &=   \sigma^2  \left[  d g'(\|x\|^2) + 2  \|x\|^2 g''(\|x\|^2)  +    \beta \|x\|  \psi'(\|x\|)   g'(\|x\|^2)  \right]. \label{eq:S-term-1}  
\end{align}
Since $Z \overset{\dd}{=} - Z$ and $Z \in A^c$ if and only if $-Z \in A^c$,  we have 
\begin{equation}
    \frac{1}{2}\E \left[   S(Z)   \ind_{A^c}(Z)\right] =  \E \left\{ \left[ G(x + \sigma Z) - G(x) \right] p(x, Z) \ind_{A^c}(Z) \right\} \eqqcolon \mathrm{Tail}(x).  \label{eq:S-term-2}
\end{equation}  
By~\eqref{eq:decomp-four-terms},~\eqref{eq:S-term-1} and~\eqref{eq:S-term-2}, it only remains to bound 
\begin{equation}
    \mathrm{Rem}(x) \coloneqq  -\frac{1}{2} \E \left[ \tilde{S}(Z)  \ind_{A^c}(Z)\right]+  \frac{1}{2}\E \left[ \bar{\mathcal{R}}(Z)  \ind_A(Z)\right]. 
\end{equation}

To bound $\E[ \tilde{S}(Z)\ind_{A^c}(Z) ] $, we use~\eqref{cond:f} and~\eqref{cond:psi} and  $|x^\top z| \leq \|x\|\|z\|$ to get   
\begin{equation}\label{eq:bound-F1-F2-D}
    |F_1(z)|  \leq  2 \fc_1 \frac{\sigma  \|z\|}{\|x\|^{2 \omega + 1}}, \qquad 
    |F_2(z)| \leq  ( \fc_1  + 2 \fc_2) \frac{\sigma^2  \|z\|^2}{\|x\|^{2 \omega + 2}}, \qquad 
    |D(z)|  \leq \beta \fb_1  \frac{\sigma \|z\|}{\|x\|}, 
\end{equation}
which hold  for every $z$. Hence, 
\begin{equation}
    \left| \tilde{S}(z) \right| \leq C_0 \frac{\sigma^2  \|z\|^2}{\|x\|^{2 \omega + 2}}, \qquad C_0 = 2 ( \fc_1  + 2 \fc_2) + 2 \beta \fb_1 \fc_1 . 
\end{equation}
Since on $A^c$, $\|Z\|^2 \leq  2 \sigma \|Z\|^3 /  \|x\|  $, 
\begin{equation}
\E \left[ \left|\tilde{S}(Z)\right|  \ind_{A^c}(Z)\right] \leq \frac{C_0 \sigma^2}{\|x\|^{2 \omega + 2}} \E \left[ \|Z\|^2 \ind_{A^c}(Z)\right]     \leq \frac{2 C_0 \sigma^3  }{\|x\|^{2 \omega + 3}}  \E\left[ \|Z\|^3 \ind_{A^c}(Z) \right]. 
\label{eq:S-term-4}
\end{equation}
For $\bar{\mathcal{R}}(Z) $, we observe that by Lemma~\ref{lm:taylor} and~\eqref{eq:bound-F1-F2-D},  the three components of $\mathcal{R}(z)$ appearing in~\eqref{eq:def-total-remainder}   all have order $O(\sigma^3 \|z\|^3 / \|x\|^{2\omega + 3})$ on $A$. Indeed, a direct calculation using $|p(x, z)| \leq 1$ yields 
\begin{align}
      |\mathcal{R}(z)| \ind_A(z)   &\leq C_1 \frac{ \sigma^3 \|z\|^3 }{\|x\|^{2 \omega + 3}} \ind_A(z), 
\end{align}
where $C_1$ is given by~\eqref{eq:C1}. 
Consequently, 
\begin{equation}
    \E\left[ \left| \bar{\mathcal{R}}(Z)\right| \ind_A(Z) \right] = \E\left[  \left| \mathcal{R} (Z) +  \mathcal{R} (-Z) \right| \ind_A(Z)  \right]    \leq \frac{2 C_1 \sigma^3 }{\|x\|^{2 \omega + 3}} \E\left[ \|Z\|^3 \ind_A(Z) \right].  \label{eq:S-term-3}
\end{equation}
Combining~\eqref{eq:S-term-4} and~\eqref{eq:S-term-3}, we obtain the claimed bound for $ |\mathrm{Rem}(x)|$. Comparing the terms in $C_0$ with those in $C_1$ shows that $C_1 > C_0$ whenever $\fc_1 \leq \fc_2$. 
\end{proof}  

\begin{lemma}\label{lm:taylor}
    Under the setting of Proposition~\ref{prop:taylor}, if $\|x\| > R$ and $\sigma \|z\| \leq \|x\|/2$, then the terms in the Taylor expansions \eqref{eq:taylor-G} and \eqref{eq:taylor-pi} satisfy  
    \begin{align}  
        \left|\mathcal{R}_G(z)\right| &\leq \frac{4^{\omega + 2} (3 \fc_2 + 2 \fc_3 )  }{3 } \frac{ \sigma^3 \|z\|^3 }{\|x\|^{2 \omega + 3}},  \\   
    \left|\mathcal{R}_\pi(z)\right| &\leq  \left[  \frac{ \beta^2 \fb_1^2 }{2}   + 2 \beta (\fb_1 + \fb_2)  \right] \frac{\sigma^2 \|z\|^2}{\|x\|^2},  \label{eq:first-order-pi} \\
    \left| p(x, z) - 1 \right| & \leq 2 \beta \fb_1 \frac{\sigma \|z\|}{\|x\|}. 
    \end{align}
\end{lemma}

\begin{proof}
We treat $x, z$ as fixed throughout the proof. 
Consider the Taylor expansion of $G$ first. Define $\varphi(\zeta) = G(x + \zeta z) = g(\|x + \zeta z \|^2)$. Letting $\rho = x+ \zeta z$, we obtain from a routine calculation that 
\begin{align}
    \varphi'(\zeta) &= 2 ( \rho^\top z) g'( \|\rho\|^2 ), \label{eq:varphi-d1} \\ 
    \varphi''(\zeta) &= 2 \|z\|^2  g'( \|\rho\|^2 ) + 4 ( \rho^\top z )^2 g''( \|\rho\|^2 ), \label{eq:varphi-d2} \\ 
    \varphi'''(\zeta) &= 12 \|z\|^2  ( \rho^\top z) g''( \|\rho\|^2 ) + 8 ( \rho^\top z )^3 g'''( \|\rho\|^2 ).  
\end{align} 
The conditions $\|x\| > R$ and $\sigma \|z\| \leq \|x\|/2$ ensure that $\|\rho \| \geq \|x\| / 2 > R / 2$ for $\zeta \in [0, \sigma]$, and thus using~\eqref{cond:f} and  $ |\rho^\top z | \leq \|\rho\|\|z\|$, we get  
\begin{align} 
    \left| \varphi'''(\zeta) \right| &\leq  (12  \fc_2 + 8 \fc_3) \frac{\|z\|^3}{   \|\rho\|^{ 2 \omega + 3} } \leq 2^{2 \omega + 3} (12  \fc_2 + 8 \fc_3) \frac{ \|z\|^3}{  \|x\|^{ 2 \omega + 3}}. \label{eq:varphi-d3}
\end{align}
By the second-order Taylor expansion, 
\begin{equation}\label{eq:taylor-order3}
    G(x + \sigma z) - G(x) = \varphi(\sigma) - \varphi (0)   = \sigma \varphi'(0) + \frac{\sigma^2}{2} \varphi''(0)   + \frac{\sigma^3}{6} \varphi'''(\zeta^*) 
\end{equation}
for some $ \zeta^* \in [0, \sigma]$.  This yields the identity~\eqref{eq:taylor-G} with $F_1(z) =\sigma \varphi'(0),$ $ F_2(z) =  \sigma^2 \varphi''(0) / 2$ and $\mathcal{R}_G(z) = \sigma^3  \varphi'''(\zeta^*) / 6$. The bound on $\left|  \mathcal{R}_G(z)  \right|$ follows from~\eqref{eq:varphi-d3}. 
 
For the acceptance probability, we begin by using Taylor expansion to approximate the log density ratio $\beta[ \psi( \|x + \zeta z\|) - \psi(\|x\|)]$. 
Define $\phi(\zeta) = \psi( \|x + \zeta z\|)$. Then, letting $\rho = x + \zeta z$, we have 
\begin{align}
    \phi'(\zeta) &= \frac{ \rho^\top z}{ \|\rho\|} \psi'(\|\rho\|), \\
    \phi''(\zeta) &= \left[ \|z\|^2  - \frac{(\rho^\top z)^2}{\|\rho\|^2}  \right] \frac{\psi'(\|\rho\|)}{\|\rho\|} + \frac{ (\rho^\top z)^2}{ \|\rho\|^2} \psi''(\|\rho\|). 
\end{align}
Using~\eqref{cond:psi}, $ |\rho^\top z | \leq \|\rho\|\|z\|$ and $\|\rho \| \geq \|x\| / 2$, we get  
\begin{align}
    \left| \phi'(\zeta) \right|  &\leq \fb_1 \frac{ \|z\| }{ \|\rho\|} \leq 2 \fb_1 \frac{ \|z\| }{ \|x\|}, \label{eq:phi-d1} \\
    \left| \phi''(\zeta) \right|   &\leq (\fb_1  + \fb_2) \frac{ \|z\|^2 }{ \|\rho\|^{2}}  \leq 4(\fb_1  + \fb_2) \frac{ \|z\|^2 }{ \|x\|^{2}}. \label{eq:phi-d2}  
\end{align}
For  some $\zeta^\star_1 \in [0, \sigma]$, the first-order Taylor expansion yields 
\begin{align}
    \log \frac{\pi_\beta(x + \sigma z) }{\pi_\beta(x)}  = D(z) + \frac{\beta \sigma^2}{2} \phi''(\zeta^\star_1), \qquad D(z) = \beta \sigma \phi'(0).  
\end{align}  
Since the function $y \mapsto 1 \wedge e^y$ has Lipschitz constant $1$, and $0 \leq e^y - 1 - y \leq y^2 / 2$ for $y \leq 0$,  
\begin{align}
  |\mathcal{R}_\pi(z)| &=  \left| p(x, z) - 1 - [0 \wedge D(z)] \right| \\ 
    &\leq \left| p(x, z) - \left[ 1 \wedge e^{D(z)} \right] \right|  
    + \left| \left[ 1 \wedge e^{D(z)} \right] - 1 - [0 \wedge D(z)] \right|  \\
    & \leq \left| \frac{\beta \sigma^2}{2} \phi''(\zeta^\star_1) \right| + \frac{D(z)^2}{2}. 
\end{align}  
Applying~\eqref{eq:phi-d1} with $D(z) = \beta \sigma \phi'(0)$ yields $|D(z)| \leq \beta \sigma \fb_1 \|z\| / \|x\|$, and then applying~\eqref{eq:phi-d2}  yields the claimed bound on  $|\mathcal{R}_\pi(z)|$. To prove the bound on $|p(x, z) - 1|$, note that since $1 - e^y \leq - y$ for $y \leq 0$,  
\begin{align}
    |p(x, z) - 1| \leq \left| \log \frac{\pi_\beta(x + \sigma z) }{\pi_\beta(x)}  \right| =  \beta \sigma \left| \phi'(\zeta^\star_0)  \right| 
\end{align}
for some $\zeta^\star_0 \in [0, \sigma]$. Apply~\eqref{eq:phi-d1} to conclude the proof. 
\end{proof}
 
We now prove  two drift conditions for $\cT$. 
We always assume  $\beta >  d / \gamma$ so that $\pi_\beta$ is well defined.

\begin{proposition}\label{prop:drift}
Suppose Assumptions~\ref{ass:pi} and~\ref{ass:kappa} hold. For $R>0$, define 
\begin{equation}\label{eq:def-V-explicit}
V(x) =  1 + (R/2)^{-\nu} - \max\{ \|x\|, R/2 \}^{- \nu}, \text{ where } \nu = \frac{1}{2}\min \left\{1,  \;    \gamma(1 - \beta) - 2  \right\}. 
\end{equation}  
If $\gamma > d + 2$ and $d / \gamma < \beta < (\gamma - 2 ) / \gamma$, then we can choose sufficiently large $R > 0$ such that 
\begin{equation} 
\frac{ (\cT V)(x) - V(x) }{V(x)} \leq - \frac{\sigma^2 \nu}{ 8 \|x\|^{\nu + 2} }, \text{ whenever } \|x\| > R. 
\end{equation} 
\end{proposition}

\begin{proof}[Proof of Proposition~\ref{prop:drift}]
Let $K$ be given by~\eqref{eq:cond-psi}, and $R_0 \geq 2$ be the smallest value  such that  
\begin{equation}\label{eq:def-R0}
   u \, \psi'(u) \leq - \gamma + 1, \quad | u \,  \psi'(u) | \leq 2 \gamma, \quad 
   | u^2 \,  \psi''(u) | \leq K \gamma, 
   \text{ for any }  u \geq R_0 /2. 
\end{equation} 
Assumption~\ref{ass:pi} ensures $R_0 < \infty$.  
The conditions $ \beta < (\gamma - 2)/\gamma$  and $R_0 \geq 2$ imply that 
\begin{align}
0 < \nu \leq \frac{1}{2}, &\quad  1 \leq V(x) < 2, \quad \psi'(u) < 0 \text{ for any }  u \geq R_0 /2. 
%\text{ and } & Z_\beta \coloneqq \int_{\bbR^d} \pi(x)^\beta \dd x < \infty.  \label{eq:def-Z} 
\end{align} 

Write $V(x) = g(\|x\|^2)$. Then,   $g(u) = 1 + (R/2)^{-\nu} - u^{-\nu / 2}$ for $u > (R/2)^2$.   
A direct calculation verifies that for $R \geq R_0$, the conditions assumed in Proposition~\ref{prop:taylor} are satisfied with 
\begin{align}
    \fb_1 = 2 \gamma, \qquad \fb_2 = K \gamma, \qquad
    \omega = \frac{\nu}{2}, \notag \\
     \fc_1 &= \frac{\nu}{2}, \qquad 
    \fc_2 = \frac{\nu(\nu+2)}{4}, \qquad 
    \fc_3 = \frac{\nu(\nu+2)(\nu+4)}{8}.
\end{align}
Hence, applying Proposition~\ref{prop:taylor} yields 
\begin{equation}\label{eq:V-drift}
    (\cT V)(x) - V(x) = \frac{ \sigma^2 \nu}{ 2 \|x\|^{\nu + 2}} \left\{    d - \nu - 2 + \beta \|x\|  \psi'(\|x\|)    \right\} + \mathrm{Tail} (x) +  \mathrm{Rem} (x),    
\end{equation} 
 where the remainder term satisfies, for any $x$ with $\|x\| > R$, 
\begin{align}
    \left|\mathrm{Rem} (x)\right| &\leq   \frac{C_1 \sigma^3 \m_3 }{ \|x\|^{\nu + 3 } }, \qquad   C_1 =   \frac{ 2^{ \nu + 2}} {3} \nu(\nu + 2) (\nu + 7)  + 2\beta \gamma \nu (\nu + 5 + K + \beta \gamma).  \label{eq:remainder-V}
\end{align} 
For the tail term, we only need an upper bound on $V(x + \sigma Z) - V(x)$. When $\|x + \sigma z\| > \|x\|$, differentiating   $u \mapsto u^{-\nu}$ and applying the mean value theorem yields 
\begin{equation}
    | V(x + \sigma z) - V(x) | =  \|x\|^{-\nu} - \|x + \sigma z\|^{-\nu} \leq \nu \frac{ \|x + \sigma z\| - \|x\|  }{\|x\|^{\nu + 1} } \leq \frac{\nu \sigma \|z\|}{ \|x\|^{\nu + 1} }. 
\end{equation}
When $\|x + \sigma z\| < \|x\|$,  $V(x + \sigma z) - V(x) \leq 0$. Hence, since $p(x, z) \leq 1$, 
\begin{align}
     \mathrm{Tail}(x) &\leq  \E \left[  \frac{\nu \sigma \|Z\|}{ \|x\|^{\nu + 1} }  \ind_{A(x)^c}(Z)\right] \leq \frac{4 \nu \sigma^3}{ \|x\|^{\nu + 3}} \E \left[ \|Z\|^3 \ind_{A(x)^c}(Z) \right] \leq  \frac{4 \nu \sigma^3 \m_3}{ \|x\|^{\nu + 3}}, 
\end{align}
where in the second step we use $\sigma \|z\| \geq \|x\|/2$ on $A^c(x)$.
So, using $\|x\| > R$,  we get 
\begin{equation}\label{eq:V-remainder-2}
     \mathrm{Tail}(x) + \left|\mathrm{Rem} (x)\right|  \leq   \frac{ (C_1 + 4 \nu) \sigma^3 \m_3 }{ R \|x\|^{\nu + 2} },  
\end{equation} 
and using $\nu \leq 1/2$, $\beta \gamma < \gamma - 2$ and $\gamma > 3$, we get 
\begin{equation}
    C_1 + 4 \nu \leq [40 + 2(\gamma - 2)( K + 4 + \gamma)]\nu  \leq \frac{(K + 24)\gamma^2 \nu }{4}. 
\end{equation}  

Now we choose 
\begin{equation}\label{eq:choice-R}
    R \geq \max\left\{ R_0, \,  ( K + 24 ) \gamma^2 \sigma  \m_3 \right\}. 
\end{equation}  
By our assumption on $R_0$, we have 
\begin{equation}
    d - \nu - 2 + \beta \|x\| \psi'(\|x\|)  \leq
    d -  \nu - 2 + \beta (1 - \gamma) 
    = ( d - \beta \gamma ) - \nu - (2 - \beta) < -1, 
\end{equation}
since $d / \gamma < \beta < 1$ and $ \nu > 0$.  
It follows from~\eqref{eq:V-drift},~\eqref{eq:V-remainder-2} and~\eqref{eq:choice-R} that 
\begin{equation} 
  (\cT V)(x) - V(x)  \leq  - \frac{\sigma^2 \nu}{  2 \|x\|^{\nu + 2} }
  \left[1 -  \frac{  ( K + 24 ) \gamma^2 \sigma  \m_3 }{2 R  } \right] \leq  - \frac{\sigma^2 \nu}{ 4 \|x\|^{\nu + 2} }. 
\end{equation} 
Since $V < 2$, the claimed drift condition then follows. 
\end{proof}

\begin{proposition}\label{prop:drift-log}
Suppose Assumptions~\ref{ass:pi} and~\ref{ass:kappa} hold. Define 
\begin{equation}\label{eq:def-V-log}
   V_{\log}(x) = \log(1 + \|x\|^2).  
\end{equation}   
Assume $\gamma > 0$ and $\beta >  d / \gamma$.  
There exists some constant $R > 0$ such that 
\begin{equation} 
0 > (\cT V_{\log} )(x) - V_{\log}(x)  \geq  - \frac{  2 (  \beta \gamma + 2  ) \sigma^2 }{    1 + \|x\|^{2}  }, \text{ whenever } \|x\| > R. 
\end{equation} 
\end{proposition}

\begin{proof}
  Let   $ R \geq R_0$ with $R_0$ as defined in the proof of Proposition~\ref{prop:drift}; see~\eqref{eq:def-R0}. Since $V_{\log}(x) = g(\|x\|^2)$ with   $g(u) = \log(1 + u)$, the conditions assumed in Proposition~\ref{prop:taylor} are satisfied with 
\begin{align} 
    \fb_1 = 2 \gamma, \qquad  \fb_2 = K \gamma, \qquad 
    \omega = 0,  \qquad 
     \fc_1 &= 1, \qquad 
    \fc_2 = 1, \qquad 
    \fc_3 = 2.
\end{align}
Since $g'(u) = (1 + u)^{-1}$ and $g''(u) = -(1 + u)^{-2}$, Proposition~\ref{prop:taylor} yields 
\begin{equation}\label{eq:V-log-drift}
    (\cT V_{\log})(x) - V_{\log}(x) =  \frac{ \sigma^2  }{ 1 + \|x\|^{2}} \left[    d + \beta \|x\|  \psi'(\|x\|) - \frac{2 \|x\|^2 }{1 + \|x\|^2}    \right] + \mathrm{Tail}(x) + \mathrm{Rem} (x),      
\end{equation} 
 where the remainder term satisfies, for any $x$ with $\|x\| > R$, 
\begin{align}
    \left|\mathrm{Rem} (x)\right| &\leq   \frac{C_1 \sigma^3 \m_3 }{ \|x\|^{3} }, \qquad 
    C_1=  \frac{112}{3}  +   4 \beta \gamma (K + 5 + \beta \gamma ).  \label{eq:C0-log}
\end{align} 
For the tail term, since $V_{\log}(x + \sigma z) - V_{\log}(x) \geq - V_{\log}(x)$ and $p(x, z) \leq 1$, we can bound it by 
\begin{align}
    \mathrm{Tail}(x) \geq - V_{\log}(x)  \E \left[  \ind_{A(x)^c}(Z) \right] \geq - \log(1 + \|x\|^2) \frac{8 \sigma^3 \m_3}{\|x\|^3}.  \label{eq:tail-log}
\end{align} 
Meanwhile, since $V_{\log}(x + \sigma z) - V_{\log}(x)  > 0$ only when $\|x + \sigma z\| > \|x\|$ and 
\begin{equation}\label{eq:grad-V-log}
    \|\nabla V_{\log}(x) \| = \frac{2 \|x\|}{1 + \|x\|^2}\leq \frac{2 }{\|x\|}, 
\end{equation}
we also have 
\begin{equation}\label{eq:upper-tail-log}
    \mathrm{Tail}(x) \leq \E \left[ \frac{ 2 \sigma \|Z\| }{ \|x\| } \ind_{A(x)^c}(Z)  \right] \leq \frac{8 \sigma^3 \m_3}{ \|x\|^3 }. 
\end{equation}
Combining~\eqref{eq:V-log-drift},~\eqref{eq:C0-log} and~\eqref{eq:tail-log} and using $u \psi'(u) > - 2\gamma$ for $u > R$, we get 
\begin{equation}
     (\cT V_{\log})(x) - V_{\log}(x) \geq  - \frac{ \sigma^2 (2 \beta \gamma + 2 - d) }{ 1 + \|x\|^{2}}  + o\left( \frac{ \sigma^3 \m_3 }{\|x\|^2} \right).     
\end{equation}
Since $\beta \gamma > d$,  the claimed lower bound on $ (\cT V_{\log})(x) - V_{\log}(x)$ follows by choosing sufficiently large $R$. 
To prove  $ (\cT V_{\log})(x) - V_{\log}(x) < 0$,  it suffices to use~\eqref{eq:upper-tail-log} and note that $R$ can be chosen sufficiently large so that $\beta u \psi'(u) < - \beta \gamma + 1$ whenever $u > R$. 
\end{proof}

\newpage 

\subsection{Proof of Theorem~\ref{th:ergodic}}  \label{sec:proof-ergodic}

\begin{proof}%[of Theorem~\ref{th:ergodic}]
Recall that $\cT$ is the transition kernel of the sampling chain in Algorithm~\ref{alg:it-rwm}. 
Let $\sA$ denote the generator of the continuous-time chain $(Y_t)_{t \geq 0}$~\citep[Chap. 4.2]{EthierKurtz1986}, which is given by 
\begin{align}\label{eq:def-gen-A}
    (\sA g)(x) &=   \int_{\bbR^d} \left[ g(y) - g(x) \right] \cK(x, \dd y)  = \frac{1}{w(x)}  [ (\cT g)(x) - g(x) ], \\ 
\text{ where }  w(x) &= \frac{\pi(x)}{\pi_\beta(x)} =  Z_\beta \pi(x)^{1 - \beta}, \qquad 
Z_\beta \coloneqq \int_{\bbR^d} \pi(x)^\beta \dd x. 
\label{eq:def-Z} 
    %&=  \frac{1}{w(x)} \int_{\bbR^d} \kappa(z) \left[ g(x + \sigma z) - g(x) \right]    \left[ 1 \wedge \frac{\pi(x + \sigma z)^\beta}{ \pi(x)^\beta} \right] \dd z,  
\end{align}  
The condition $\beta > d / \gamma$ ensures  $Z_\beta < \infty$.
We use the two drift conditions constructed in Propositions~\ref{prop:drift} and~\ref{prop:drift-log} to prove the sufficiency and necessity parts of Theorem~\ref{th:ergodic}  respectively.

First, let $V$ be as given in Proposition~\ref{prop:drift} with sufficiently large $R>0$, and assume $\beta < (\gamma - 2)/\gamma$. Then, 
\begin{equation}\label{eq:sA-drift-0}
   \frac{ (\sA V)(x) }{V(x)}\leq - \frac{\sigma^2 \nu}{8 \|x\|^{\nu+2} w(x)} = - \frac{\sigma^2 \nu}{8  Z_\beta \|x\|^{\nu+2}  \pi(x)^{1-\beta}}, \text{ whenever } \|x\| \geq R, 
\end{equation}
where $\nu = \min\{ 1, \,   \gamma(1 - \beta) - 2 \}/2 \in (0, 1/2]$.   
Define 
\begin{equation}\label{eq:def-delta}
    \delta \coloneqq \gamma (1 - \beta) - 2  - \nu \geq  \frac{\gamma(1 - \beta) - 2}{2} > 0. 
\end{equation}
Assumption~\ref{ass:pi} implies that $\pi(x)$ is  regularly varying  and can be written as $\pi(x) \propto \|x\|^{-\gamma} L( \|x\| )$ for some slowly varying function $L$~\citep{nair2022fundamentals}. Consequently, 
\begin{equation}
    \|x\|^{\nu + 2}  \pi(x)^{1 - \beta}  \propto  \|x\|^{\nu + 2 - \gamma(1 - \beta)} L(\|x\|)^{1 - \beta } \leq \|x\|^{ -\delta } L(\|x\|)^{1 - \beta }. 
\end{equation} 
Since $L$ is slowly varying, this further implies that $\|x\|^{\nu + 2}  \pi(x)^{1 - \beta} \rightarrow 0$ as   $\|x\| \rightarrow \infty$.   Therefore, by~\eqref{eq:sA-drift-0}, we can select some $\alpha_R > 0$ independent of $x$ such that 
\begin{equation}\label{eq:drift-sA}
(\sA V)(x) \leq - \alpha_R V(x), \text{ whenever } \|x\| > R. 
\end{equation}
%Indeed, we can pick $ \alpha_R > 0$   such that $\alpha_R \uparrow \infty$ as $R \uparrow \infty$.    
Since the function $V$ is bounded, the uniform ergodicity of $(Y_t)_{t \geq 0}$ follows from~\citet[Thm. 5.2]{down1995exponential}, once we can verify that the set $\cB_R = \{x \colon \|x\| \leq R \}$ is petite for $(Y_t)_{t \geq 0}$. 
To prove $\cB_R$ is petite, it suffices to note that (1) the average holding time at any $x$ with $\|x\| \leq R$ is bounded by 
\begin{align*}
 0< \inf_{x \colon \|x\| \leq R} Z_\beta \pi(x)^{1 - \beta}  \leq   w(x) \leq  \sup_{x \colon \|x\| \leq R} Z_\beta \pi(x)^{1 - \beta}  < \infty, 
\end{align*}
since $\psi$ is continuous, and that (2) $\cB_R$ is small for the discrete-time chain $\cT$ under Assumptions~\ref{ass:pi} and~\ref{ass:kappa}.   
This proves the first part of Theorem~\ref{th:ergodic}. 

For the necessity part of Theorem~\ref{th:ergodic}, we  first note that under the assumption $\pi(x) \geq c_L \|x\|^{-\gamma}$ for all sufficiently large $\|x\|$,  $Z_\beta = \infty$ whenever $\beta \leq d / \gamma$,  in which case $(Y_t)_{t \geq 0}$ is not well defined.  
To show that $\beta< (\gamma - 2)/\gamma$ is also necessary, we  use the drift function $V_{\log} = \log(1 + \|x\|^2)$. By Proposition~\ref{prop:drift-log}, for any $\beta > d / \gamma$, the following drift condition holds for $\sA$: 
\begin{equation}
   0 >  (\sA V_{\log})(x) \geq  -   \frac{ 2 (\beta \gamma + 2) \sigma^2 }{  (1+\|x\|^2) \, Z_\beta \, \pi(x)^{1 - \beta} }, \text{ whenever } \|x\| \geq R. 
\end{equation}
We can assume $R \geq R_L$ so that $\pi(x) \geq c_L \|x\|^{-\gamma}$  whenever $\|x\| > R$. 

Now we prove $\beta< (\gamma - 2)/\gamma$ is a necessary condition for the uniform ergodicity of $(Y_t)_{t \geq 0}$ by contradiction. 
Assume $\beta \geq (\gamma - 2) / \gamma$, which implies $2 -(1 - \beta) \gamma \geq 0$. Since $(1+\|x\|^2)   \pi(x)^{1 - \beta} \geq c_L^{1-\beta}  > 0$ whenever $\|x\| > R$, there exists some $\alpha_{\log} > 0$ such that 
\begin{equation}\label{eq:drift-A-log}
  0 >  (\sA V_{\log})(x) = \frac{(\cT V_{\log})(x) - V_{\log}(x) }{w(x)}   
    \geq  - \alpha_{\log}, \text{ whenever } \|x\| \geq R.  
\end{equation} 
Assume further that  $(Y_t)_{t\geq 0}$ is uniformly ergodic. By applying Theorem 16.2.2 of~\citet{meyn2012markov} to the $T$-skeleton chain (i.e., the discrete-time chain $(\tilde{Y}_i)_{ i \geq 1}$ with $\tilde{Y}_i = Y_{iT}$), we obtain $\sup_{x \in \bbR^d}\E_x [ \tau_R ] < \infty$, where 
\begin{equation}\label{eq:def-tau-R}
    \tau_R = \inf\{  t \geq 0 \colon   \|Y_t\|  \leq R  \}. 
\end{equation} 
Let $(X_i)_{i \geq 1}$ denote the embedded discrete-time chain with transition kernel $\cT$, which has initial value $X_1 = Y_0 = x$, and denote the holding time at $X_i$ by $W_i$. Then, 
\begin{equation}\label{eq:tau-decomp}
\tau_R = \sum_{i=1}^{N - 1} W_i, \text{ where } N = \inf\{ j \geq 1 \colon \|X_j\| \leq R \}. 
\end{equation}
Since $(X_i)_{i \geq 1}$ is the RWM chain, it is Harris recurrent, and thus $N < \infty$ almost surely. 
Moreover, since $\E_x[ W_i \mid (X_j)_{j \geq 1}] = w(X_i)$, applying the law of total expectation yields  
\begin{equation}\label{eq:tau-expect-decomp}
    \E_x[\tau_R ] =  \E_x \left[ \sum_{i=1}^{N - 1} W_i \right]
    =  \E_x \left[ \sum_{i=1}^{N - 1} w(X_i) \right]. 
\end{equation} 
Let $\cF_n = \sigma(X_1, \dots, X_n)$ and consider the discrete-time stopped martingale $(M_{n \wedge N})_{n \geq 1}$, where 
\begin{align}
    M_n &= \sum_{i=1}^{n - 1} [ V_{\log}(X_{i + 1})  - (\cT V_{\log})( X_i)   ] \\
    &=  V_{\log}(X_n) - V_{\log}(x) - \sum_{i=1}^{n -1} \left[  (\cT V_{\log})(X_{i}) -  V_{\log} (X_{i}) \right], \qquad n \geq 1.  \label{eq:mart-increment}
\end{align}
Taking expectation on both sides for the stopped martingale and using~\eqref{eq:drift-A-log} and~\eqref{eq:tau-expect-decomp}, we get 
\begin{align}
  0 = \E_x\left[ M_{n \wedge N} \right] %&= \E_x\left\{  V_{\log}(X_{n \wedge N}) - V_{\log}(x) - \sum_{i=1}^{ ({n \wedge N}) -1} \left[  (\cT V_{\log})(X_{i}) -  V_{\log} (X_{i}) \right]  \right\}  \\
   &\leq \E_x[V_{\log}(X_{n \wedge N})] - V_{\log}(x) + \alpha_{\log}\E_x \left[ \sum_{i=1}^{ ({n \wedge N}) -1} w(X_i) \right] \\
   &\leq \E_x[V_{\log}(X_{n \wedge N})] - V_{\log}(x) + \alpha_{\log}\E_x \left[ \tau_R \right]. 
   \label{eq:mart-increment-expect} 
\end{align}   
If 
\begin{equation}\label{eq:L1-conv}
    V_{\log}(X_{n \wedge N}) \rightarrow V_{\log} (X_N) \text{ in } L^1, 
\end{equation}
then, by using $N < \infty$ almost surely, we get 
\begin{equation}
    \log(1 + R^2) \geq \E_x \left[    V_{\log}( X_N )  \right] = \lim_{n \rightarrow \infty}  \E_x[V_{\log}(X_{n \wedge N})]  \geq V_{\log}(x) - \alpha_{\log}  \E_x[ \tau_R ]. 
\end{equation}
And this would yield the contradiction, 
\begin{equation}
    \sup_{x: \|x\| \geq R} \E_x[ \tau_R ] \geq \sup_{x: \|x\| \geq R} \frac{V_{\log}(x) - \log(1 + R^2)}{\alpha_{\log}} = \infty, 
\end{equation} 
which implies that if $\beta \geq (\gamma - 2) / \gamma$, $(Y_t)_{t\geq 0}$ cannot be uniformly ergodic. 

So it only remains to prove~\eqref{eq:L1-conv}. We begin by computing the quadratic variation of $M_{n \wedge N}$: 
\begin{align} 
 \langle M \rangle_{n \wedge N} &\coloneqq \sum_{k=1}^{(n \wedge N) - 1} \E_x \left[  (M_{k + 1} - M_{k})^2 \mid \cF_{k} \right] = \sum_{k=1}^{(n \wedge N) - 1} q(X_{k}),  \label{eq:qv}  \\
\text{where } q(x)&= \int_{\bbR^d} [V_{\log}(y) - (\cT V_{\log})(x)]^2 \cT(x, \dd y ) \leq \int_{\bbR^d} [V_{\log}(y) -  V_{\log} (x)]^2 \cT(x, \dd y ). 
\end{align}
Fix $x \in \bbR^d$ with $\|x\| >  R$, and write $y = x + \sigma z$.  Using~\eqref{eq:grad-V-log} and $| \|y\| - \|x\| | \leq \sigma \|z\|$, we obtain 
\begin{equation}
   \left| V_{\log}(y) - V_{\log}(x) \right| \ind_{ \{ \|y\| \geq \|x\| / 2 \} }  \leq  \frac{2 \sigma \|z\|}{ \|x\| / 2 } = \frac{4 \sigma \|z\|}{ \|x\| }. 
\end{equation}
When $\|y\| < \|x\| / 2$, we must have $\sigma \|z\| > \|x\| / 2$, and therefore 
\begin{equation}
       \left| V_{\log}(y) - V_{\log}(x) \right| \ind_{ \{ \|y\| < \|x\| / 2 \} }  \leq \log (1 + \|x\|^2) \ind_{ \{ \sigma \|z\| > \|x\| / 2 \} }. 
\end{equation} 
As in the proof of Proposition~\ref{prop:drift-log}, letting $Z \sim \kappa$ and using $\E \|Z\|^2 = d$ and $p(x, Z) \leq 1$, we find that 
\begin{align}
    q(x)  &\leq  \E \left\{ \left[  V_{\log}(x + \sigma Z) - V_{\log}(x) \right]^2 \right\}  \\
    & \leq \frac{16 \sigma^2 }{ \|x\|^2 } \E   \|Z\|^2   + 
    \log^2 (1 + \|x\|^2) \E   \left[   \ind_{\{  \sigma \|Z\| > \|x\| / 2\} } \right]  
    \\
    & \leq \frac{16 \sigma^2 d}{ \|x\|^2 } + \log^2(1 + \|x\|^2) \frac{8 \sigma^3 \m_3}{\|x\|^3}.  
\end{align} 
Since  $\beta \geq (\gamma - 2)/\gamma$ implies $\|x\|^2 w(x) \geq Z_\beta c_L^{1 - \beta} > 0$, there exists some constant $C > 0$ independent of $x$  such that $q(x) \leq C w (x)$ whenever $\|x\| > R$. 
%\begin{equation}\label{eq:compare-q-w}
%q(x) \leq C w (x), \text{ whenever } \|x\| > R. 
%\end{equation}  
Hence, using~\eqref{eq:tau-expect-decomp} and~\eqref{eq:qv} we get 
\begin{equation}
 \E_x\left [  M_{n \wedge N}^2 \right]  = \E_x\left[ \langle M \rangle_{n \wedge N}  \right] \leq C  \E_x \left[ \sum_{k=1}^{(n \wedge N) - 1} w(X_k) \right]  \leq C\, \E_x[\tau_R]  < \infty. 
\end{equation}
So $(M_{n \wedge N})_{n \geq 1}$ is an $L^2$-bounded martingale, and by Doob's inequality, $\overline{M} = \max_{n \leq N} |M_n|$ is square integrable. 
By~\eqref{eq:mart-increment} and using $(\cT V_{\log})(x) - V_{\log}(x) < 0$, we get 
\begin{align}
     V_{\log}(X_{n \wedge N}) 
    &=  M_{n \wedge N} + V_{\log}(x) +  \sum_{i = 1}^{ (n \wedge N) - 1} [ (\cT V_{\log})(X_i)  - V_{\log}(X_i)] 
    \leq \overline{M} + V_{\log}(x).
\end{align} 
Dominated convergence theorem implies that $V_{\log}(X_{n \wedge N}) \rightarrow V_{\log}(X_N)$ in $L^2$, completing the proof. 
\end{proof}  
%\cite{darling2008differential}

\newpage 

\subsection{Proof of Corollary~\ref{coro:light}} \label{sec:proof-light}

\begin{proof}%[of Corollary~\ref{coro:light}]
Fix any $\beta \in (0, 1)$  and choose $\gamma > 0$ such that $d/\gamma < \beta < (\gamma - 2)/\gamma$.  
Let $V$ be as given in Proposition~\ref{prop:drift}. 
We still use $\cT$ to denote the transition kernel of the RWM chain with stationary density $\pi_\beta$, which satisfies   
\begin{align}\label{eq:int-1}
( \cT V - V )(x) = 
\int_{\bbR^d} \kappa(z) [V(x + \sigma z) - V(x)] \left[ 1 \wedge \frac{\pi(x + \sigma z)^\beta}{ \pi(x)^\beta} \right] \dd z. 
\end{align}
Define a reference probability density  
\begin{equation}
\bar{\pi}(x) = \exp\left\{ \bar{\psi}(\|x\|)\right\} \propto (1 + \|x\|)^{-\gamma}. 
\end{equation}
Let $\overline{\cT}$ denote the RWM transition kernel with stationary density proportional to $\bar{\pi}^\beta$, i.e., 
\begin{align}\label{eq:int-2}
(\overline{\cT} V - V )(x) = 
\int_{\bbR^d} \kappa(z) [V(x + \sigma z) - V(x)] \left[ 1 \wedge \frac{\bar{\pi}(x + \sigma z)^\beta}{ \bar{\pi}(x)^\beta} \right] \dd z. 
\end{align} 
The main idea of the proof is to show that the  drift under $\cT$ is at least as strong as that under $\overline{\cT}$, and the time-change by $w(x)^{-1}$ only strengthens this effect in the tails. 

For each of the integrals in~\eqref{eq:int-1} and~\eqref{eq:int-2}, we partition the domain of integration into $E$ and $E^c$, where 
$E \coloneqq  \{z \colon  \|x + \sigma z\| > \|x\| \}$. 
Since both $\pi$ and $\bar{\pi}$ are radially monotone decreasing, 
on $E^c$ the acceptance probabilities of both chains equal one, and thus the two integrals coincide. 
% \begin{equation}
% E^c = \{z \colon  \pi(x + \sigma z) \geq \pi(x)\}
%  = \{z \colon  \bar{\pi}(x + \sigma z) \geq \bar{\pi}(x)\}, 
% \end{equation}
So it remains to compare the integrals on $E$. 
Write $\pi(x) = \exp\{ \psi(\|x\|)\}$ where $\psi(u) = C - u^a$  for some constant $C \in \bbR$. Then,
\begin{align*}
\psi'(u) = - a u^{a - 1}, \quad \bar{\psi}'(u) = - \frac{\gamma}{1 + u}. 
\end{align*}
It is easy to verify that there exists some $\bar{R} < \infty$ such that for all $u \geq \bar{R}$, $\psi'(u) \leq \bar{\psi}'(u)$. 
Hence, for any $u \geq \bar{R}$ and $h > 0$, 
\begin{equation}
\psi(u + h) - \psi(u) \leq \bar{\psi}(u + h) - \bar{\psi}(u). 
\end{equation} 
It follows that, whenever $\|x\| \geq \bar{R}$, 
\begin{align*}
\left[ \frac{\pi(x + \sigma z)^\beta}{\pi(x)^\beta}  - 1 \right]\ind_E 
%&= \left[ \frac{\pi(x + \sigma z)^\beta}{\pi(x)^\beta}  - 1 \right]\ind_E   \\
&\leq  \left[ \frac{ \bar{\pi}(x + \sigma z)^\beta}{\bar{\pi}(x)^\beta}  - 1 \right]\ind_E,  
\end{align*}
and consequently,  
\begin{align*}
(\cT V - V)(x) \leq (\overline{\cT} V - V )(x).  
\end{align*}

Since $\bar{\pi}$ satisfies Assumption~\ref{ass:pi}, the drift condition given in Proposition~\ref{prop:drift} holds for $\overline{\cT} $ and thus also for $\cT$,  as long as $R$ is sufficiently large. Hence, 
\begin{equation} 
\frac{ (\sA V)(x) }{V(x)} = \frac{ (\cT V - V)(x)}{w(x) V(x)}
\leq - \frac{\sigma^2 \nu}{ 8 \|x\|^{\nu + 2} w(x)}, \text{ whenever } \|x\| > R,
\end{equation} 
where $w(x) = \pi(x)/\pi_\beta(x) \propto \exp\{ - (1-\beta)\|x\|^a\}$.  
Since $1/w(x)$ grows faster than $\|x\|^{\nu + 2}$ for any $a > 0, \nu \geq 0$, there exists some $\alpha_R > 0$ such that $(\sA V)(x) \leq - \alpha_R V(x)$ whenever $\|x\| > R$.
%Moreover, we can choose the drift coefficient $\alpha_R$ in~\eqref{eq:drift-sA} such that it grows to infinity as $R$ increases, since $\exp\{ (1-\beta)\|x\|^a\}$ grows faster than $\|x\|^{\nu + 2}$ for any $\nu \geq 0$. Hence, 
%the drift condition~\eqref{eq:drift-sA}  holds. 
By~\citet[Thm. 5.2]{down1995exponential}, $(Y_t)_{t \geq 0}$ is uniformly ergodic. 
\end{proof}

\newpage

\subsection{Proof of Theorem~\ref{th:decomp}}\label{sec:proof-decomp}
\begin{proof} 
Fix some bounded $f$ with $\Pi(f) = 0$. Define $J(\dd x, \dd y) \coloneqq \cQ(x, \dd y) \Pi(\dd x) $. By the reversibility assumption on $\cQ$, 
$J$ is symmetric, i.e., $J(\dd x, \dd y ) = J(\dd y, \dd x)$, and it follows that 
\begin{align}
&\quad     \int_{\bbR^d \times \bbR^d}      \frac{V(x) - V(y)}{V(x)} f(x)^2     J(\dd x, \dd y)  \\
&=     \int_{\bbR^d \times \bbR^d}      \frac{V(y) - V(x)}{V(y)} f(y)^2     J(\dd y, \dd x) \\
&= \frac{1}{2}\int \left\{       \frac{ V(x) - V(y)}{V(x)} f(x)^2 +        \frac{ V(y) - V(x)}{V(y)} f(y)^2    \right\} J(\dd x, \dd y) \\
&\leq \frac{1}{2}\int  [f(y) - f(x)]^2 J(\dd x, \dd y) \\ 
&= \cE_{\cQ}(f, f),  \label{eq:symmetry-J}
\end{align}
where the second last step follows from the elementary inequality $a f(x)^2 + a^{-1} f(y)^2 - 2 f(x) f(y) = [ a^{1/2} f(x) - a^{-1/2} f(y) ]^2 \geq 0$. 
The integral on the left-hand side exists and is finite due to the given drift condition, and we can bound it from below by 
\begin{align}
     \int_{\bbR^d \times \bbR^d}      \frac{V(x) - V(y)}{V(x)} f(x)^2     J(\dd x, \dd y)
     &=   \int_{\bbR^d}\left\{ - \int_{\bbR^d}\left[ \frac{V(y)}{V(x)} - 1 \right] \cQ(x, \dd y) \right\}  f(x)^2  \Pi(\dd x)\\
     & \geq   \int_{\bbR^d} \left[ \alpha_R \ind_{\cB_R^c}(x) - b_R \ind_{\cB_R}(x) \right] f(x)^2 \Pi(\dd x) \\
     & = \alpha_R \Pi(f^2 \ind_{\cB_R^c}) - b_R \Pi(f^2 \ind_{\cB_R}). 
\end{align}
Hence, it follows from~\eqref{eq:symmetry-J} that 
\begin{equation}\label{eq:f2-bound-Bc}
\Pi(f^2 \ind_{\cB_R^c}) \leq \frac{b_R}{\alpha_R} \Pi(f^2 \ind_{\cB_R})  + \frac{1}{\alpha_R}\cE_{\cQ}(f, f). 
\end{equation}

Since $\Pi(f) = 0$, we have 
$\Pi(f \ind_{\cB_R}) = - \Pi(f \ind_{\cB_R^c})$.  
Letting $\mu_R(f) = \Pi(f \ind_{\cB_R}) / \Pi(\cB_R)$ 
denote the conditional mean  under $\Pi_R$ and 
using the Cauchy-Schwarz inequality, we find that 
\begin{equation}\label{eq:cs}
   p_R^2 \, \mu_R(f)^2 = \Pi(f \ind_{\cB_R} )^2  = \Pi(f \ind_{\cB_R^c} )^2 \leq \Pi(f^2 \ind_{\cB_R^c}) \Pi(\cB_R^c) = \Pi(f^2 \ind_{\cB_R^c}) (1 - p_R). 
\end{equation}
Let $\Var_{\Pi_R}(f) = \int_{\cB_R} [f(x) - \mu_R(f) ]^2 \Pi_R(\dd x)$ denote the conditional variance under $\Pi_R$. We get 
\begin{align}
 \Pi( f^2 \ind_{\cB_R}) &= p_R [ \Var_{\Pi_R}(f) +  \mu_R(f)^2 ]  \leq   p_R  \Var_{\Pi_R}(f) + \frac{1 - p_R}{p_R}   \Pi(f^2 \ind_{\cB_R^c}) \\
 & \leq  \frac{1}{\lambda_R} \cE_{\cQ}(f, f)  + \frac{1 - p_R}{p_R}   \Pi(f^2 \ind_{\cB_R^c}). 
 \label{eq:f2-bound-Bc-2}
\end{align}
The last step follows from the definition of the spectral gap $\lambda_R$
\begin{equation}
    \lambda_R = \inf_{f \colon \Var_{\Pi_R}(f) > 0} \frac{ \cE_{R}(f, f) }{\Var_{\Pi_R}(f)},
\end{equation}
where the Dirichlet form $\cE_{R}(f, f)$ satisfies 
\begin{align*}
         \cE_R(f, f) &= \frac{1}{2}\int_{\cB_R \times \cB_R} [ f(y) - f(x) ]^2 \Pi_R(\dd x) \cQ_R(x, \dd y)  \\ 
         &= \frac{1}{2 p_R}\int_{\cB_R \times \cB_R} [ f(y) - f(x) ]^2 \Pi(\dd x) \cQ(x, \dd y)  \leq \frac{\cE_{\cQ}(f, f)}{p_R}. 
\end{align*} 
Combining~\eqref{eq:f2-bound-Bc} and~\eqref{eq:f2-bound-Bc-2} and rearranging the terms, we get 
\begin{equation}\label{eq:f2-bound-Bc-3} 
   \frac{  \alpha_R p_R - b_R(1 - p_R) }{\alpha_R p_R}  \Pi(f^2 \ind_{\cB_R^c})   \leq  \left( \frac{b_R}{\alpha_R \lambda_R} + \frac{1}{\alpha_R} \right) \cE_{\cQ}(f, f). 
\end{equation}
By assumption, $\alpha_R p_R \geq 2b_R(1 - p_R)$. So we can now use~\eqref{eq:f2-bound-Bc-2} and~\eqref{eq:f2-bound-Bc-3} to get 
\begin{align}
    \Var_\Pi(f) &= \Pi( f^2 \ind_{\cB_R}) + \Pi( f^2 \ind_{\cB_R^c})   \\ 
    & \leq \frac{1}{\lambda_R}  \cE_{\cQ}(f, f)  +   \frac{1}{p_R} \Pi( f^2 \ind_{\cB_R^c})  \\
    & \leq \frac{1}{\lambda_R}  \cE_{\cQ}(f, f) + \frac{1 + b_R / \lambda_R}{ \alpha_R p_R - b_R(1 - p_R) }  \cE_{\cQ}(f, f) \\
    & \leq \frac{\cE_{\cQ}(f, f)}{\lambda_R}   \left[ 1  + \frac{2 (\lambda_R + b_R ) }{ \alpha_R p_R   } \right]. \label{eq:ineq-restricted-gap}
\end{align}
By a standard truncation argument, we can approximate each unbounded and centered $f \in L^2(\pi)$ using a sequence of bounded and centered functions $(f_n)_{n \geq 1}$ with $\Var_\Pi(f_n) \rightarrow \Var_\Pi(f)$ and $\cE_{\cQ}(f_n, f_n) \leq \cE_{\cQ}(f, f)$. Hence,~\eqref{eq:ineq-restricted-gap} holds for every centered $f \in L^2(\pi)$.  
The claim then follows from the definition of the spectral gap. 
\end{proof}

\newpage 

\subsection{Proof of Theorem~\ref{th:gap}} \label{sec:proof-th-gap}
\begin{proof}%[of Theorem~\ref{th:gap}]
We split the proof into multiple steps. 

\medskip
\noindent\textit{Step 1: Drift condition on $\cB_R^c$.} 
Define $\gamma = m + d$, which is the tail index of $\Pi$. 
A simple calculation shows that $\pi(x) \propto e^{\psi(\|x\|)}$, where $\psi$ satisfies
\begin{equation}
u \psi'(u) = -  \frac{\gamma u^2}{u^2 + m}, \quad u^2 \psi''(u) =   \frac{\gamma u^2 (u^2 - m) }{( u^2 + m)^2}. 
\end{equation}
Hence, Assumption~\ref{ass:pi} is satisfied with $K = 2$. Further, since $\gamma > d + 2$, we can define $R_0 = 2 \sqrt{\gamma (\gamma - d)} \leq 2 \gamma$, which satisfies $R_0 \geq 2$ and equation~\eqref{eq:def-R0} in the proof of Proposition~\ref{prop:drift}. 
Hence, we can define $V$ by 
\begin{align} 
V(x) &=  1 + (R/2)^{-\nu} - \max\{ \|x\|, R/2 \}^{- \nu}, \\
\text{ where } \nu &= \frac{1}{2}\min \left\{1,  \;    \gamma(1 - \beta) - 2  \right\} = \frac{\epsilon}{2}, \\
R &\geq   \, \max\left\{ 2\gamma, \;   26   \sigma \gamma^2 \m_3 \right\},  \label{eq:choice-R-gap}
\end{align}    
where the condition on $R$ follows from~\eqref{eq:choice-R}. 
By Proposition~\ref{prop:drift}, the following drift condition holds:
\begin{align}
      \frac{(\sA V)(x)}{V(x)} =   \int_{\bbR^d} &\left[ \frac{ V(y)}{V(x)} - 1 \right]  \cK(x, \dd y)   \leq - \frac{\sigma^2 \nu}{8 \|x\|^{\nu + 2}  w(x)}, \; \text{ if } x \in \cB_R^c.
\end{align} 
As in the proof of Theorem~\ref{th:ergodic}, we write $w(x) = Z_\beta \pi(x)^{1 - \beta}$. 
Using $1 + \|x\|^2 / m \geq \|x\|^2 / m$, we get 
\begin{equation}\label{eq:def-cdm}
  \pi(x) = \frac{c_{d, m} }{  ( m + \|x\|^2  )^{ \gamma / 2} } \leq  \frac{c_{d, m}}{\|x\|^{\gamma}}, \text{ where }  c_{d, m} =  \frac{\Gamma[ (m+d)/2 ] m^{m/2}  }{\Gamma( m/2)   \pi^{d/2} }. 
\end{equation}
Hence, in order to ensure that $ (\sA V)(x) \leq - \alpha_R V(x)$ on $\cB_R^c$, we can choose 
\begin{equation}\label{eq:gap-bound-alpha}
    \alpha_R = \frac{\sigma^2 \nu}{8 Z_\beta    c_{d, m}^{1 - \beta}  }  R^{\nu},  
\end{equation}
where we have used 
\begin{equation}
     \gamma(1 - \beta) - 2 - \nu = \epsilon - \nu = \frac{\epsilon}{2} = \nu. 
\end{equation}

\medskip
\noindent\textit{Step 2: Drift condition on $\cB_R$.}  Next, we find $b_R$ such that 
\begin{align}
    \frac{(\sA V)(x)}{V(x)}   \leq  b_R, \text{ if } x \in \cB_R. 
\end{align}
This follows directly from the boundedness of $V$. Explicitly, since  $V \in [1, 1 + (R/2)^{-\nu})$, we have 
\begin{equation}
    \frac{(\sA V)(x)}{V(x)} =  \frac{1}{w(x)}    \int_{\bbR^d} \left[ \frac{V(y)}{V(x)} - 1 \right] \cT(x, \dd y ) \leq  \frac{(R/2)^{-\nu} }{w(x)}. 
\end{equation}
Since $R > \sqrt{\gamma m}$, we have 
\begin{equation}
  \pi(x) = \frac{c_{d, m}}{\|x\|^{\gamma}} \left(1 + \frac{m}{\|x\|^2} \right)^{-\gamma/2} \geq  \frac{c_{d, m}}{\|x\|^{\gamma}} \left(1 + \frac{1}{\gamma} \right)^{-\gamma/2} \geq  \frac{c_{d, m}}{\sqrt{e} \|x\|^{\gamma}}, \text{ whenever } \|x\| \geq R.  
\end{equation} 
Hence, using $\gamma(1 - \beta) - \nu = 2 +  \nu$,  we find that the drift condition on $\cB_R$ holds with  
\begin{equation}\label{eq:gap-bound-b}
    b_R =\frac{2^\nu \sqrt{e} }{ Z_\beta  c_{d, m}^{1 - \beta}  }   R^{2 + \nu} \geq  \sup_{x\colon \|x\| \leq R} \frac{(R/2)^{-\nu} }{w(x)}  
    = \frac{(R/2)^{-\nu} }{ Z_\beta \pi(R)^{1 - \beta}}. 
\end{equation}
  
\medskip
\noindent\textit{Step 3: Bound on $p_R$.}  To bound  $p_R = \Pi(\cB_R)$,  we perform the routine polar coordinate transformation, which yields 
\begin{align}
 1 - p_R &=   \frac{\Gamma(\gamma/2)}{\Gamma( m/2)   (m\pi)^{d/2} }  \frac{2 \pi^{d/2}}{\Gamma(d/2)} \int_R^\infty   u^{d - 1} \left(1 + \frac{u^2}{ m} \right)^{-\gamma / 2} \dd u \\
 &\leq \frac{\Gamma(\gamma/2)}{\Gamma( m/2)   (m\pi)^{d/2} }  \frac{2 \pi^{d/2}}{\Gamma(d/2)} \int_R^\infty   u^{d - 1}   \frac{u^{-\gamma}}{m^{-\gamma/2}} \dd u \\ 
 & = \frac{2 \Gamma(\gamma/2) m^{m / 2 - 1} }{\Gamma( m/2)   \Gamma(d/2)  } R^{-m} \\ 
 & = \frac{ \Gamma(\gamma/2) m^{m / 2} }{\Gamma[(m+2)/2]   \Gamma(d/2)  } R^{-m}.   
\end{align}
To bound the right-hand side, we use the following standard inequalities for the gamma function~\citep[Theorem 1.6]{batir2008inequalities},  
\begin{equation}\label{eq:gamma-ineq}
  \frac{\Gamma(x)}{\Gamma(y)} \leq x^{x - y}, \quad  \Gamma(x + 1) \geq  \left( \frac{x}{e} \right)^{x}, \text{ whenever } x > y > 0.  
\end{equation}
It follows that 
\begin{equation}\label{eq:gap-bound-p}
 1 - p_R  \leq  \frac{  m^{m / 2} R^{-m}  }{ \Gamma[(m+2)/2]  } \left( \frac{\gamma}{2} \right)^{m/2}  \leq  \frac{ m^{m/ 2} }{R^m}   \left( \frac{\gamma}{2} \right)^{m/2}   \left( \frac{2e}{m} \right)^{m/2} = (\gamma e)^{m/2} R^{-m}. 
\end{equation}
Since $R \geq 2\gamma$ and $\gamma = m + d > d + 2$, one can numerically verify that~\eqref{eq:gap-bound-p} implies $p_R \geq 1/2$.

\medskip
\noindent\textit{Step 4: Bound on $\lambda_R$.} 
We now find an upper bound on $\lambda_R$ to show that it is smaller than $b_R$. By the definition of the spectral gap, we can pick any square-integrable test function $f$ and use $\cE_{R}(f, f)/\Var_{\Pi_R}(f)$ as an upper bound, where $\cE_R$ denotes the Dirichlet form associated with the kernel $\cK_R$. We choose  
\begin{equation}
f(x) = \ind_{\cB_1^+}(x), \quad 
    \cB_1^+ = \left\{ x \colon \|x\| \leq 1, \; x_1 > 0 \right\}, 
\end{equation}
where $x_1$ denotes the first coordinate of $x \in \bbR^d$. Note that $\cB_1^+ \subset \cB_R$ since $R > 1$. Then,   
\begin{align*}
     \cE_R(f, f) &= \frac{1}{2}\int_{\cB_R \times \cB_R} [f(y) - f(x)]^2 \cK_R(x, \dd y) \Pi_R(\dd x) \\
     &= \frac{1}{  p_R}\int_{x \in \cB_1^+} \int_{y \in \cB_R\setminus \cB_1^+} \frac{1}{w(x)}  \cT(x, \dd y) \Pi(\dd x) \\
     &\leq  \frac{1}{  p_R}\int_{x \in \cB_1^+}  \frac{1}{w(x)}   \Pi(\dd x), %\leq \frac{\Pi(\cB_1^+)}{p_R} \frac{1}{w_1}, 
\end{align*}  
where the second step uses the definition of $f$ and reversibility, and the third step uses $\cT(x, \bbR^d) = 1$. 
Since $w(x) = Z_\beta \pi(x)^{1 - \beta}$, $w(x)$ is monotone decreasing in $\|x\|$, and thus for $x \in \cB_1^+$, $w(x)^{-1}$ is maximized when $\|x\| = 1$. Using~\eqref{eq:def-cdm}, we therefore get  
\begin{equation}\label{eq:lambda_R_diri}
    \cE_R(f, f) \leq \frac{\Pi(\cB_1^+)}{p_R} \frac{1}{Z_\beta c_{d,m}^{1 - \beta} } \left( m + 1\right)^{ \gamma (1-\beta) / 2}
    =  \frac{\Pi_R(\cB_1^+)}{Z_\beta c_{d,m}^{1 - \beta} } \left(  m + 1\right)^{ \gamma(1 - \beta) / 2}. 
\end{equation}
Since $\cB_1^+ \subset \cB_R$ and $\Pi_R( \cB_1^+) = \Pi_R(\cB_1) / 2$ due to symmetry, we have 
\begin{equation}\label{eq:lambda_R_var}
    \Var_{\Pi_R}(f) = \Pi_R(\cB_1^+)[ 1 - \Pi_R(\cB_1^+)] \geq \frac{1}{2} \Pi_R(\cB_1^+). 
\end{equation}
Using $\gamma(1 - \beta) / 2 = 1 + \nu$ and combining~\eqref{eq:lambda_R_diri},~\eqref{eq:lambda_R_var} and~\eqref{eq:gap-bound-b}, we find that 
\begin{equation}
    \lambda_R \leq \frac{ \cE_R(f, f) }{  \Var_{\Pi_R}(f) } \leq  \frac{2}{Z_\beta c_{d,m}^{1 - \beta} } \left(  m + 1\right)^{ 1 + \nu}  
    \leq \frac{R^{ 2 + \nu}}{Z_\beta c_{d,m}^{1 - \beta} }   \leq b_R, 
\end{equation}
since $R > \gamma \geq m + 1 \geq 3$.    

\medskip
\noindent\textit{Step 5: Application of Theorem~\ref{th:decomp}.}
To apply Theorem~\ref{th:decomp}, we need to ensure that $2b_R(1 - p_R) \leq  \alpha_R p_R$, which holds if $\alpha_R \geq 4 b_R (1 - p_R)$ since $p_R \geq 1/2$.
By combining~\eqref{eq:gap-bound-alpha},~\eqref{eq:gap-bound-b} and~\eqref{eq:gap-bound-p},  we find that this condition holds when 
\begin{align*}
    \frac{\sigma^2 \nu }{2^{\nu + 5} \sqrt{e} (\gamma e)^{m/2} } R^{m - 2} \geq 1. 
\end{align*}
Since $\nu = \epsilon / 2 \in(0, 1/4)$, we have $2^{\nu + 5} \sqrt{e}  \leq 63$. So  it suffices to require that 
\begin{equation}\label{eq:choice-R-gap2}
    R \geq \left( \frac{126}{\sigma^2 \epsilon} \right)^{\frac{1}{m - 2}} (\gamma e)^{ \frac{m}{2(m-2)} }. 
\end{equation}  
Combining~\eqref{eq:choice-R-gap} and~\eqref{eq:choice-R-gap2} gives the claimed choice of $R$ in the theorem statement. 

Since $p_R \geq 1/2$ and $b_R \geq \lambda_R$, by Theorem~\ref{th:decomp}, we get 
\begin{align*}
    \mathrm{Gap}(\cK) \geq \lambda_R \left[1 + \frac{2(\lambda_R + b_R)}{ \alpha_R p_R} \right]^{-1} \geq \lambda_R \left[1 + \frac{8 b_R}{ \alpha_R} \right]^{-1},  %= \frac{\lambda_R \alpha_R}{\alpha_R + 8 b_R} 
\end{align*}
where $\lambda_R = \mathrm{Gap}(\cK_R)$. 
By~\eqref{eq:gap-bound-alpha} and~\eqref{eq:gap-bound-b} and using $\nu < 1/4$, 
\begin{equation}
    \frac{8b_R}{\alpha_R} = \frac{8 \cdot  2^{\nu + 3} \sqrt{e} }{ \sigma^2 \nu  } R^2 =  \frac{ 2^{\nu + 7} \sqrt{e} }{ \sigma^2 \epsilon } R^2 \leq \frac{251 R^2}{ \sigma^2 \epsilon}. 
\end{equation}
Condition~\eqref{eq:choice-R-gap}  implies that $R^2 > 26^2 \sigma^2$ since $\gamma > d + 2$ and $ \m_3 \geq ( \int \|z\|^2 \kappa(z) \dd z)^{3/2} = d^{3/2}$.
The claimed bound thus follows. 
\end{proof}

\newpage

\section{Numerical examples} \label{sec:sim}

\subsection{Uniform ergodicity of the continuous-time chain} \label{sec:sim-unif}

In the first numerical example, we let $\Pi$ be the multivariate $t$-distribution with $m=4$ degrees of freedom. Theorem~\ref{th:ergodic} implies that $(Y_t)_{t \geq 0}$ is uniformly ergodic if and only if 
 $$\frac{d}{d+4} < \beta < \frac{d+2}{d+4}.$$  
For comparison, we consider two choices of $\beta:  \beta_1 = (d + 2 - \epsilon) / (d + 4)$ and $\beta_2 = (d + 2 + \epsilon) / (d + 4)$ with $\epsilon = 1/4$. We let the proposal density $\kappa$ be the standard Gaussian distribution and tune the step size $\sigma$ so that the expected acceptance rate is above $0.1$ for both $\beta_1, \beta_2$. 
Then, we simulate the continuous-time Markov chain $(Y_t)_{t \geq 0}$ with initial value $Y_0 = (y, \dots, y)$, using different values of $y$. 
We simulate $10^4$ trajectories, and for each fixed $t$, we compare the distribution of $Y_t$ with $\Pi$ by computing Kolmogorov--Smirnov (KS) statistics for the first coordinate and for the radius $\|Y_t\|$. 

In Figure~\ref{fig1}, we visualize the resulting KS statistics on a logarithmic scale for $d=1, 5, 10$. For each $d$, we consider four initialization schemes where $\|Y_0\|$ grows geometrically. 
It can be seen that when $\beta =  (d + 2 - \epsilon) / (d + 4)$ (i.e., $\beta = 0.55$ for $d=1$, $\beta = 0.75$ for $d = 5$ and $\beta = 0.84$ for $d = 10$), the time needed for the KS statistic to reach a given level (say, 0.02) is not substantially affected by the starting point. 
This is consistent with the uniform ergodicity of $(Y_t)_{t \geq 0}$, which ensures that the time required to approach stationarity is uniformly controlled over all initial states. 
In contrast, when $\beta =  (d + 2 + \epsilon) / (d + 4)$ (i.e., $\beta = 0.65$ for $d=1$, $\beta = 0.81$ for $d = 5$ and $\beta = 0.875$ for $d = 10$), the chain is not uniformly ergodic, and the dependence on initialization becomes much more pronounced. In particular, the corresponding curves are more clearly separated, and the gaps between successive curves remain roughly constant, which suggests that the time required for the KS statistic to reach a certain level will keep growing with $\|Y_0\|$. 

We note that by the Dvoretzky--Kiefer--Wolfowitz bound, at significance level $0.05$, the half width of the confidence interval of each KS statistic is
\begin{equation}
\sqrt{\frac{\log(2 / 0.05) }{2 \times 10^4} } \approx 0.0136. 
\end{equation}
This explains why the curves in Figure~\ref{fig1} become unstable once the KS statistic drops below 0.02, since those values are largely driven by Monte Carlo error.

\begin{figure}[htbp]
\centering
\includegraphics[scale=.073]{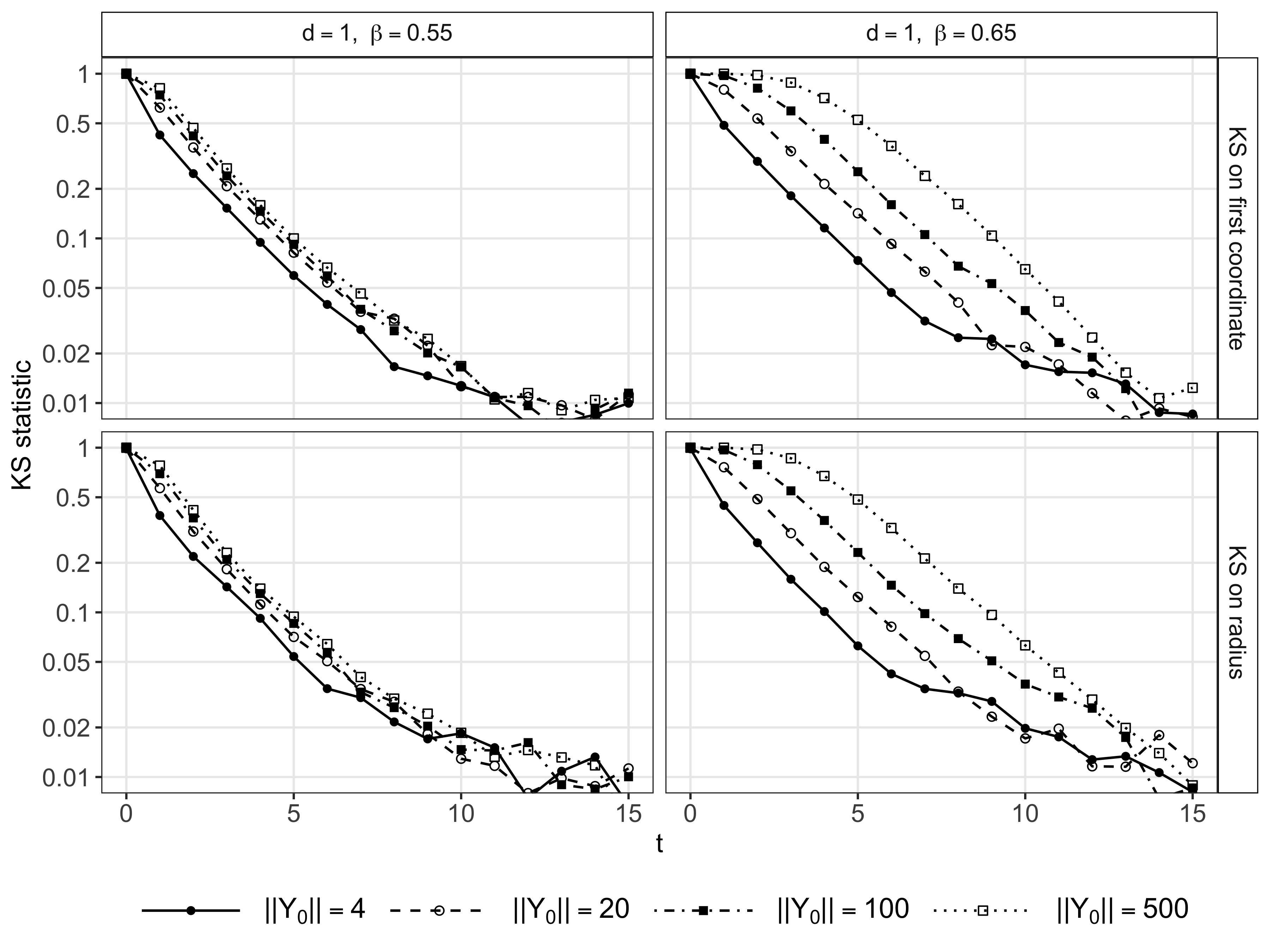}
\includegraphics[scale=.073]{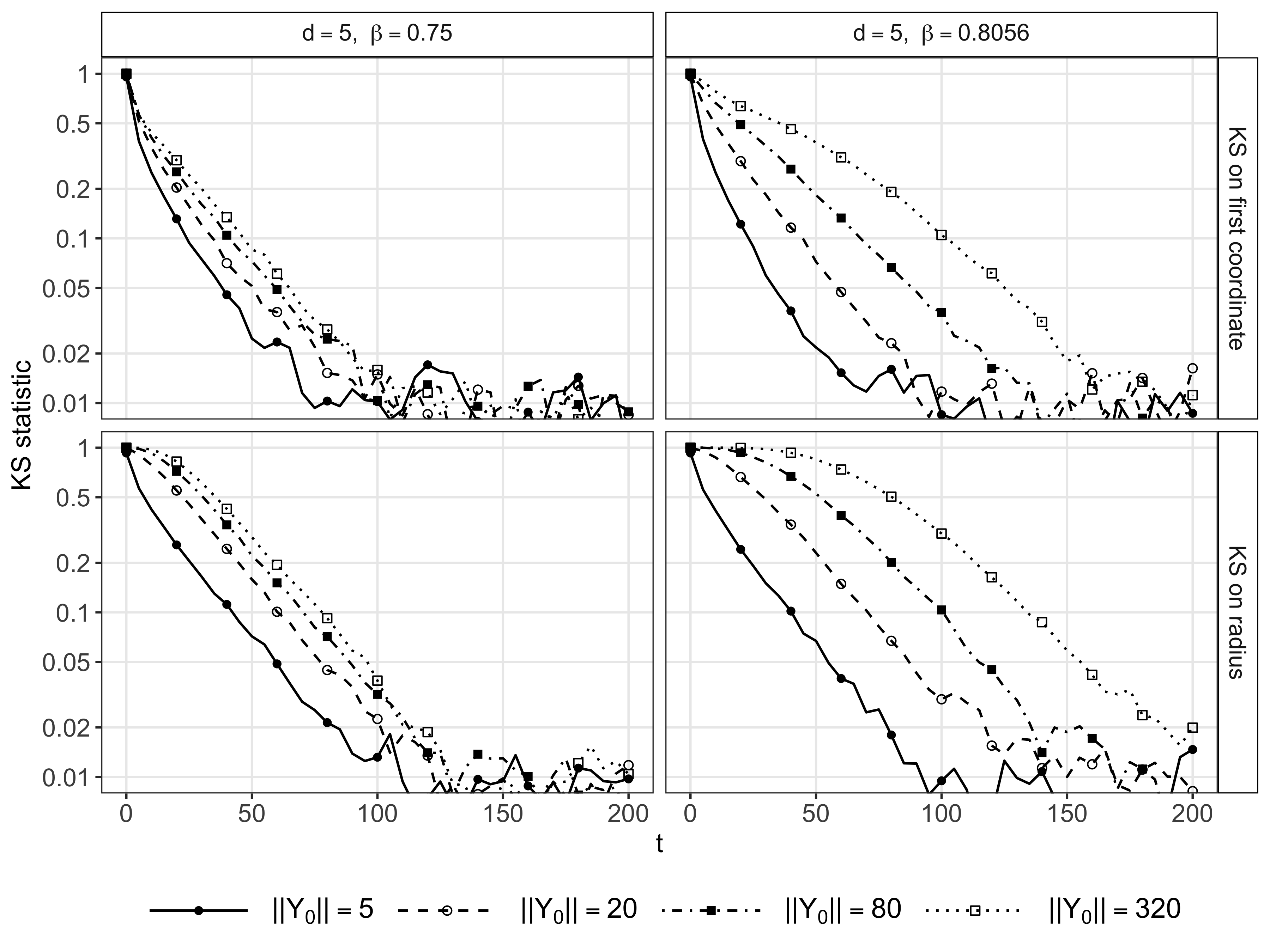}
\includegraphics[scale=.073]{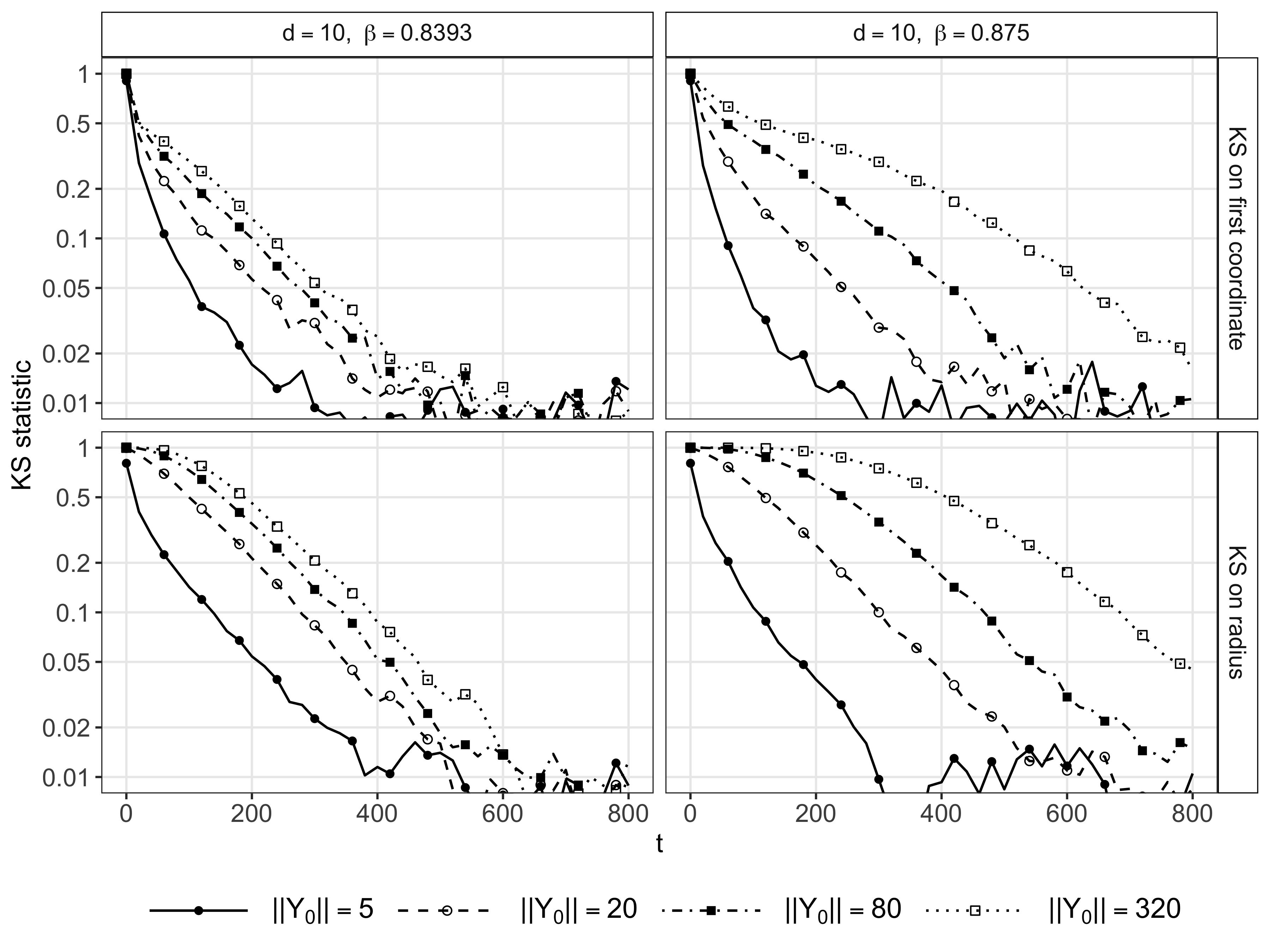}
\caption{Convergence of the continuous-time Markov chain $(Y_t)_{t \geq 0}$ with $\Pi$ being the $d$-dimensional $t$-distribution with $m=4$ degrees of freedom.  The Kolmogorov--Smirnov statistic compares the target $t$-distribution with the distribution of $Y_t$ over $10^4$ replicates. For each KS statistic, the confidence interval has width $\pm 0.014$. 
}
\label{fig1}
\end{figure}
  
\clearpage 
\newpage

\subsection{Efficiency of importance-tempered RWM}\label{sec:sim-itmh}

In the second experiment, we simulate the importance-tempered RWM sampler (i.e., Algorithm~\ref{alg:it-rwm}) and  numerically study the variance of the SNIS estimator $\widetilde{\Pi}_{\beta, n}(f)$ defined in~\eqref{eq:def-snis-beta}. 
We let $\Pi$ be the multivariate $t$ distribution with $m=5$ degrees of freedom (so that the variance of the test function $f(x) = \|x\|^2$ is finite).  
Since $\Pi_\beta$ is well defined only when $\beta > d/\gamma$ where $\gamma = d + m = d + 5$, we consider the following values of $\beta$: 
\begin{equation}
\beta = \frac{ d + 0.5k }{d + 5},   \text{ for } k = 1, 2, \dots, 10. 
\end{equation}  
When $\beta = 1$ (i.e., $k=10$), Algorithm~\ref{alg:it-rwm} reduces to the usual RWM algorithm. 
Recall that the continuous-time chain $(Y_t)_{t \geq 0}$ is uniformly ergodic when $\beta < (d + 3) / (d + 5)$, and we will denote this critical value by $\beta^*$.   
The proposal distribution $\kappa$ is Gaussian, and $\sigma$ is tuned in the same way as in the first simulation study.   
For each choice of $(d, \beta)$, we consider three initialization schemes: (i) $X_0 \sim \Pi_\beta^{\mathrm{trunc}}$, where $\Pi_\beta^{\mathrm{trunc}}$ denotes $\Pi_\beta$ conditioned on $\|X_0\|\leq 2\sqrt{d}$, (ii) $\|X_0\|=20$, (iii) $\|X_0\|=40$. In the latter two schemes, $X_0$ is deterministic with identical coordinates.   

If the sampler quickly reaches stationarity, we expect that the total importance weight $\sum_{i=1}^{n} w(X_i)$ is close to $n$, since $\int w(x) \pi_\beta(x) \dd x = 1$. Hence, when estimating the asymptotic variance of the SNIS estimator, we discard replicate runs for which $\sum_{i=1}^{n} w(X_i) < 0.05n$; in such runs, the sampler has probably not entered the high-density region, and the resulting SNIS estimates can be highly unstable.  We summarize in Table~\ref{table:fail} the number of such runs under each setting. 
The cutoff level $0.05n$ is chosen to be deliberately conservative. If the continuous-time chain $(Y_t)_{t \geq 0}$ is uniformly ergodic, then once the sampler reaches the high-density region, the importance weights accumulate rapidly, and the SNIS estimator stabilizes quickly thereafter.

In Figure~\ref{fig:beta}, we consider estimating the mean for four test functions: $f(x) = x_1$, $\|x\|$, $\|x\|^2$, and $\ind_{\{\|x\|>q_{0.75}\}}$, where $q_{0.75}$ denotes the $75$-th percentile of the radius under $\Pi$. 
For each $f$, we define the scaled estimated asymptotic variance as 
\begin{equation} \label{eq:def-av}
\text{AV}(f) = \frac{n}{\Var_\Pi(f)} \widehat{\Var} \left( \widetilde{\Pi}_{\beta, n}(f) \right), 
\end{equation}
where $n$ is the number of iterations of Algorithm~\ref{alg:it-rwm}, and the variance of $\widetilde{\Pi}_{\beta, n}(f)$ is estimated across $1,000$ independent replicate runs. No burn-in samples are discarded when computing $\widetilde{\Pi}_{\beta, n}(f)$. 
The critical  value $\beta^* = (d+3)/(d+5)$ (which equals $0.667$ for $d = 1$, $0.8$ for $d=5$, and $0.867$ for $d = 10$) is marked on the horizontal axis. We make three important remarks. 

First, in every setting, importance tempering can effectively reduce the asymptotic variance when $\beta$ is appropriately chosen. The reduction becomes highly substantial when the test function is $f(x) = \|x\|^2$ or the initial state $X_0$ is in the tails. For the test functions $f(x) = x_1$, $\|x\|$, $\ind_{\{\|x\|>q_{0.75}\}}$, the optimal value of $\beta$ that minimizes the asymptotic variance seems to be close to $\beta^*$. 
For $f(x) = \|x\|^2$, the optimal value of $\beta$ is smaller. This is because the variance of $f(x) = \|x\|^2$ is largely driven by the tail behavior (compared with the other three test functions), and a smaller $\beta$ encourages the sampler to spend more time exploring tail regions.  

Second, by comparing the estimated asymptotic variances with $\beta < \beta^*$ and those with $\beta \geq \beta^*$, we see that the initialization scheme has a much more significant impact  when $\beta \geq \beta^*$. This is again consistent with our theory about the chain $(Y_t)_{t \geq 0}$. When $(Y_t)_{t \geq 0}$ is uniformly ergodic (i.e., $\beta < \beta^*$), the SNIS estimator is largely insensitive to the initialization, since regardless of how poor the initialization is, the total importance weight accumulated in the tail regions is uniformly controlled. 
To further illustrate this phenomenon, in Figure~\ref{fig:violin} we visualize the distributions of the estimator $\widetilde{\Pi}_{\beta, n}(f)$ for $f(x) = \|x\|^2$ under different initialization schemes. When $\beta < \beta^*$, the median value and distribution of the estimator remain largely the same over all 3 initialization schemes, although the case $\|X_0\| = 40$ may produce slightly more extreme estimates. 
When $\beta > \beta^*$, the median value clearly increases over the 3 initialization schemes, and when $\|X_0\| = 40$, the entire distribution is visibly shifted upward.

Third, we recall that Figures~\ref{fig:beta} and~\ref{fig:violin} do not include problematic MCMC runs that are indicated in Table~\ref{table:fail}. Such runs yield highly unstable estimates and thus can increase the estimated asymptotic variance drastically. It can be seen from Table~\ref{table:fail} that such runs typically occur when $\beta$ is small, in which case the sampler has a much heavier-tailed stationary distribution and is therefore more likely to get stuck in the tails. 
For a fixed $\beta$, such failed runs occur more frequently when the initial state lies far in the tails. 
This does not contradict the uniform ergodic theory for the chain $(Y_t)_{t \geq 0}$: the uniform ergodicity only implies that the total importance weight accumulated in the tails is bounded uniformly, while the number of iterations of Algorithm~\ref{alg:it-rwm} required to leave the tails still grows with $\|X_0\|$ and can be unbounded. 
Overall, this simulation study suggests a trade-off in the choice of $\beta$, and values slightly less than $\beta^*$ appear to be preferable.

\begin{table}[htbp]
\centering
\caption{Number of removed replicates with total importance weight less than $0.05n$ out of $1,000$ replicate runs. 
Columns corresponding to $\beta > (d+3)/(d+5)$ are omitted because the removal counts are all zero.} \label{table:fail}
\begin{tabular}{llcccccc}
\toprule
$d$ & initialization & $\frac{d+0.5}{d+5}$ & $\frac{d+1}{d+5}$ & $\frac{d+1.5}{d+5}$ & $\frac{d+2}{d+5}$ & $\frac{d+2.5}{d+5}$ & $\frac{d+3}{d+5}$ \\
\midrule
\multirow{3}{*}{$1$}
& $\Pi_{\beta}^{\mathrm{trunc}}$ & 0 & 0 & 0 & 0 & 0 & 0 \\
& $\|X_0\|=20$                  & 5 & 1 & 0 & 1 & 0 & 0 \\
& $\|X_0\|=40$                  & 4 & 1 & 0 & 0 & 0 & 0 \\
\midrule
\multirow{3}{*}{$5$}
& $\Pi_{\beta}^{\mathrm{trunc}}$ & 1 & 0 & 0 & 0 & 0 & 0 \\
& $\|X_0\|=20$                  & 6 & 1 & 0 & 0 & 0 & 0 \\
& $\|X_0\|=40$                  & 13 & 6 & 0 & 1 & 0 & 0 \\
\midrule
\multirow{3}{*}{$10$}
& $\Pi_{\beta}^{\mathrm{trunc}}$ & 1 & 1 & 0 & 0 & 0 & 0 \\
& $\|X_0\|=20$                  & 5 & 3 & 0 & 1 & 1 & 0 \\
& $\|X_0\|=40$                  & 32 & 19 & 4 & 1 & 1 & 0 \\
\bottomrule
\end{tabular}
\end{table}

\begin{figure}[htbp]
\centering
\includegraphics[scale=.08]{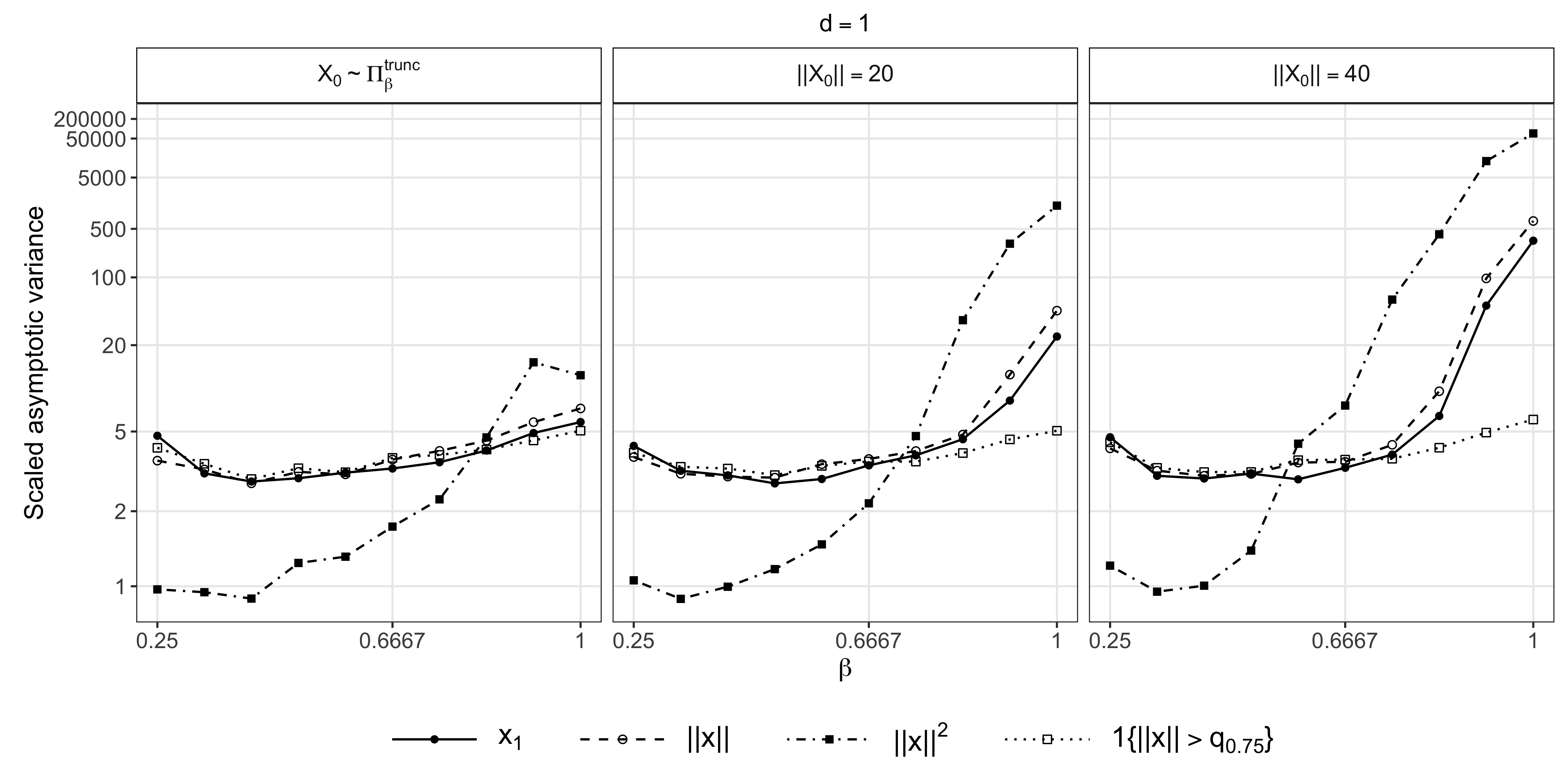}
\includegraphics[scale=.08]{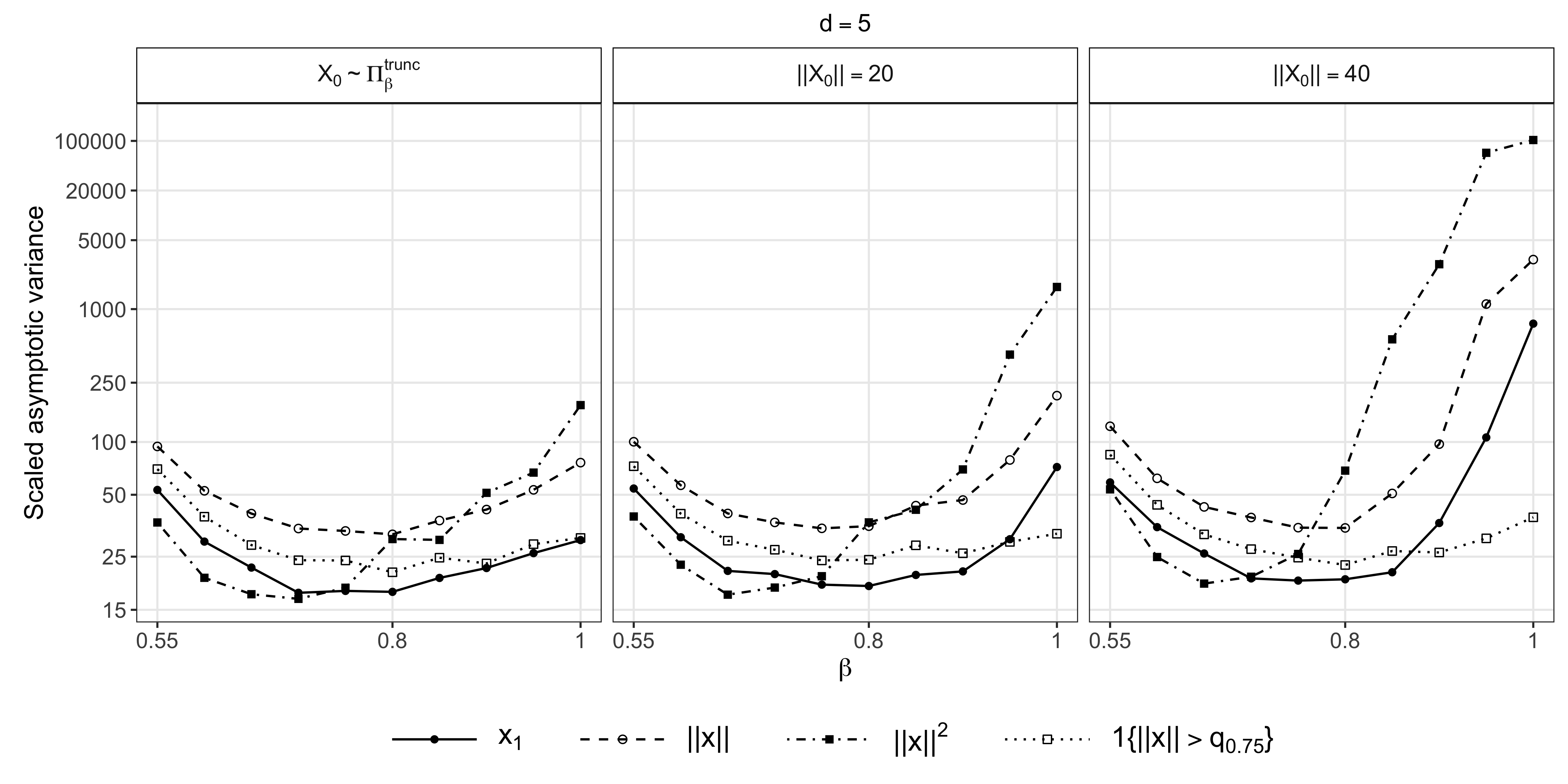}
\includegraphics[scale=.08]{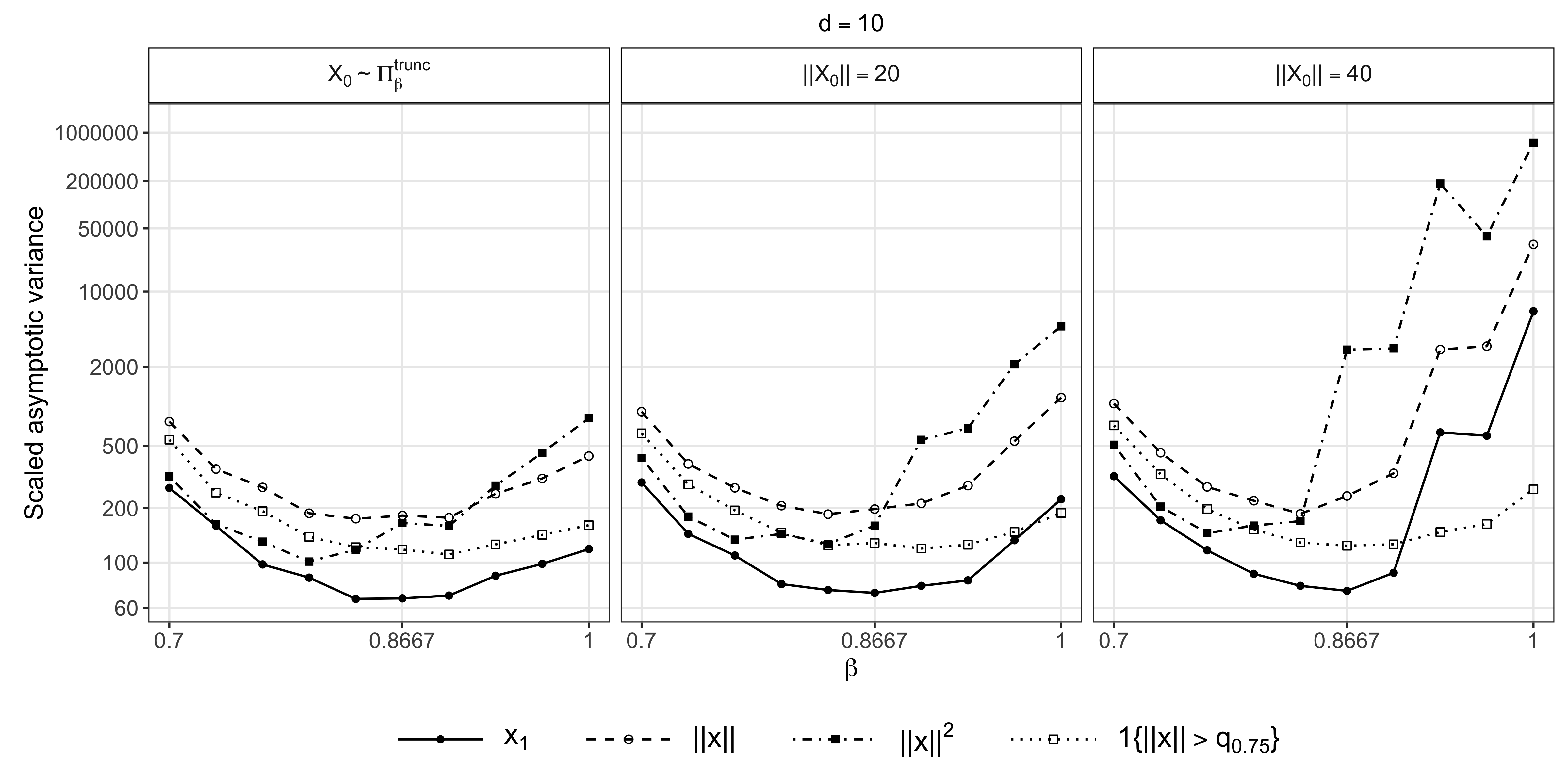}
\caption{Simulation of the importance-tempered RWM sampler under three initialization schemes.
Horizontal axis: $\beta$; the middle tick label marks the critical value $(d+3)/(d+5)$ for uniform ergodicity of $(Y_t)_{t \geq 0}$. 
Vertical axis: asymptotic variance of the SNIS estimator, estimated from $1,000$ replicates and scaled by $\Var_\Pi(f)$. 
Number of iterations $n = 5,000$ for $d=1$, $  n = 40,000$ for $d = 5$, and $n = 10^5$ for $d = 10$.   } \label{fig:beta}
\end{figure}

\begin{figure}[htbp]
\centering
\includegraphics[scale=.096]{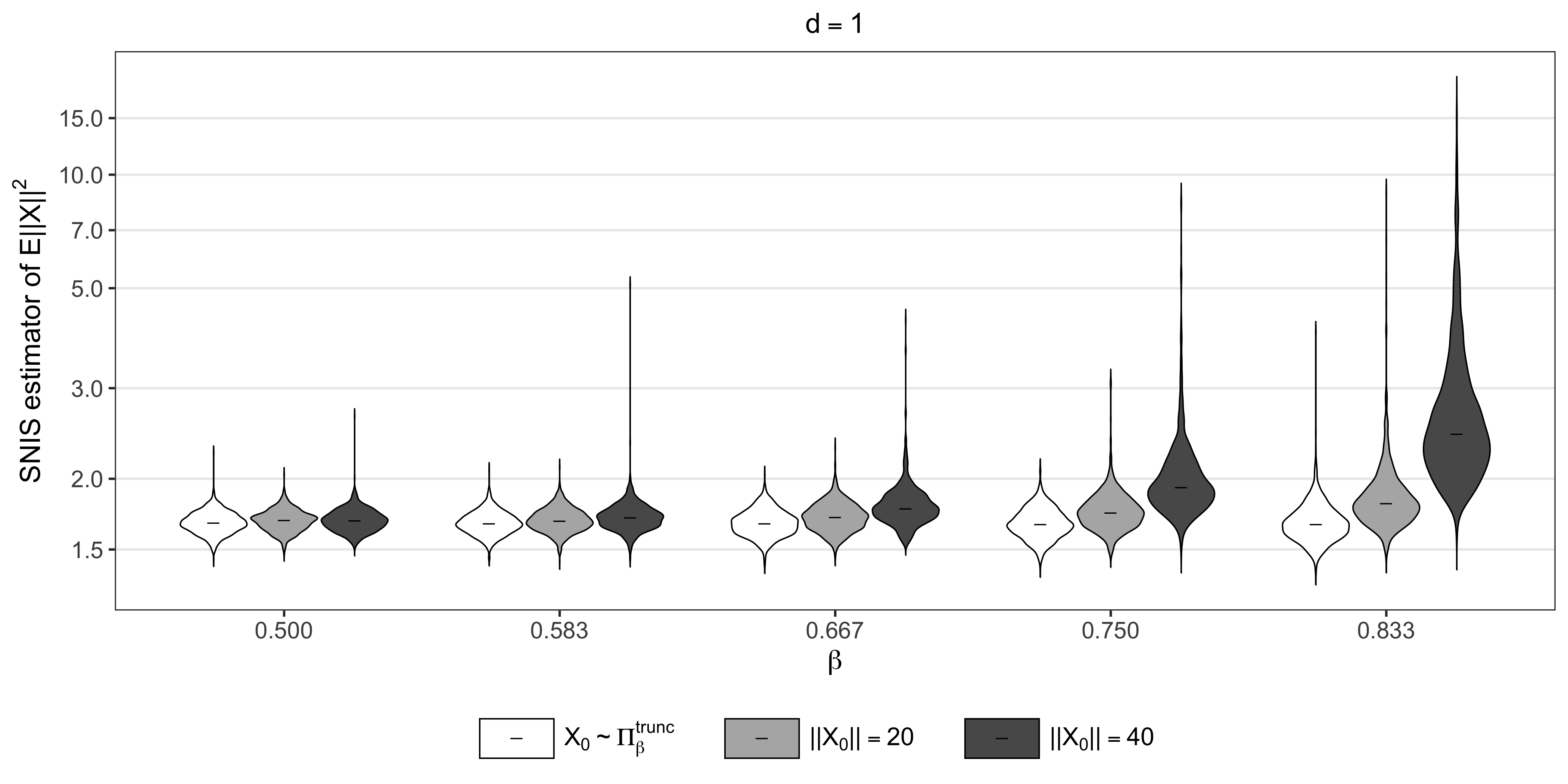}
\includegraphics[scale=.096]{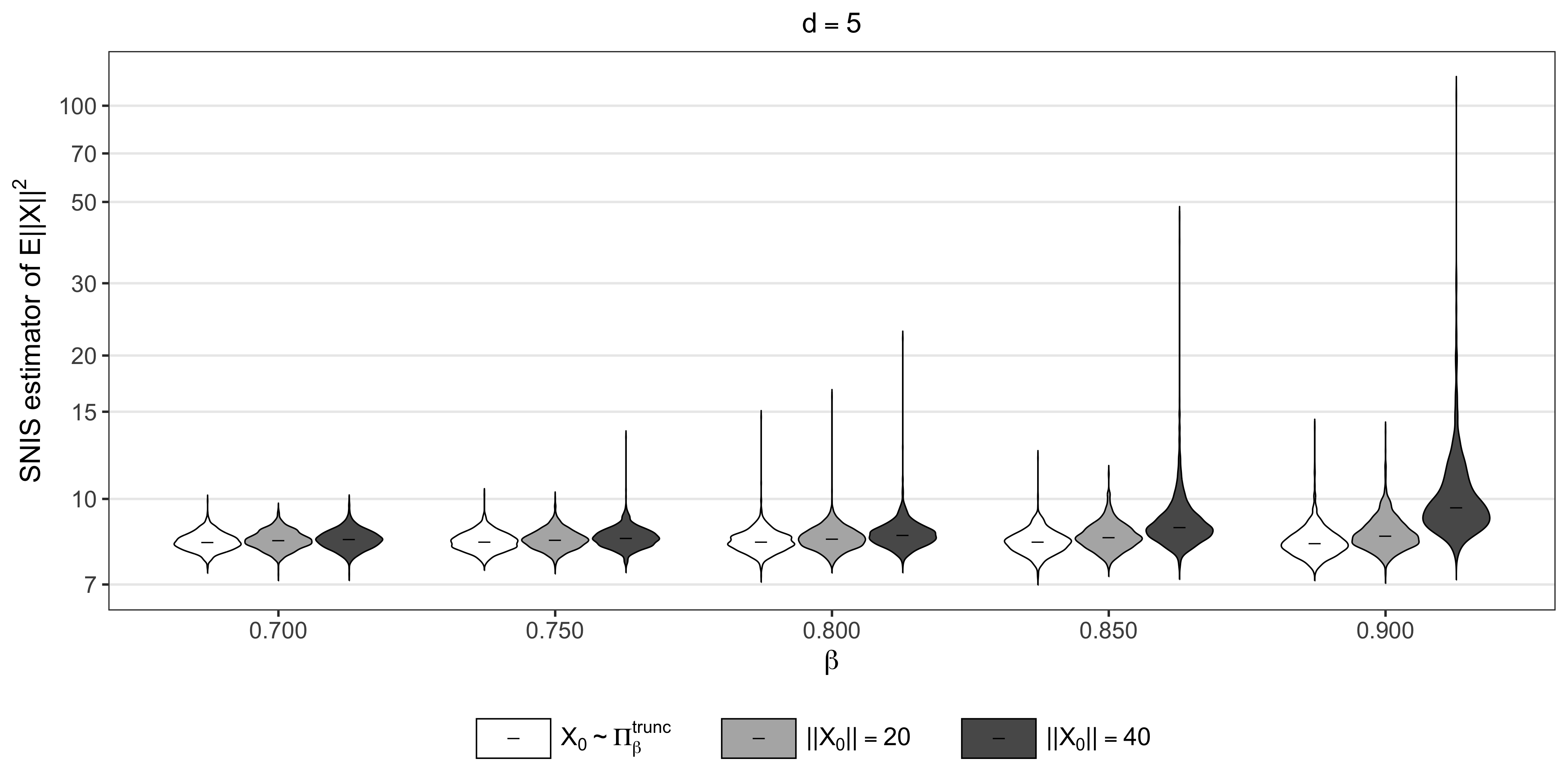}
\includegraphics[scale=.096]{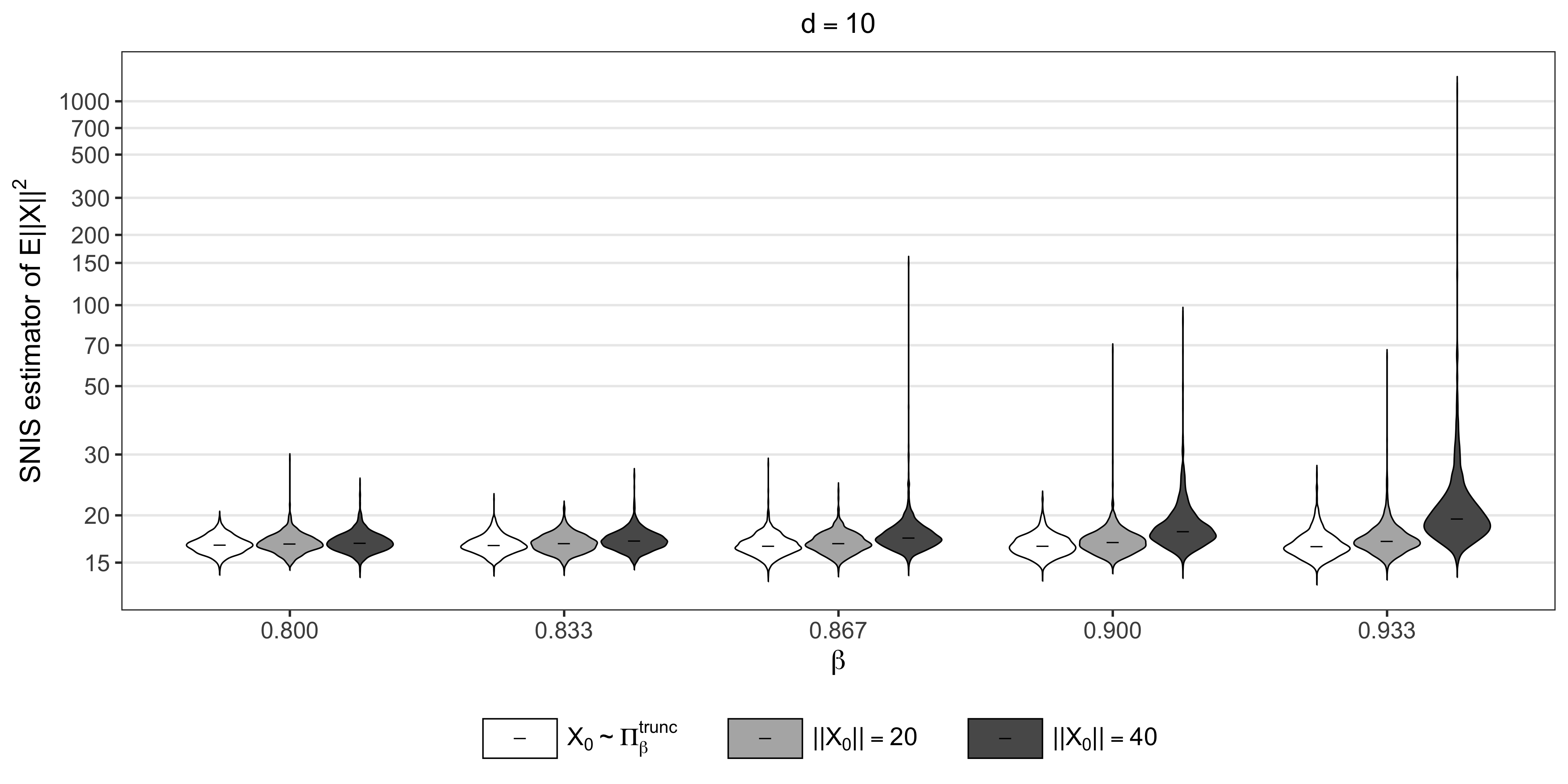}
\caption{Distribution of the SNIS estimator for $f(x) = \|x\|^2$ over 1,000 runs of the importance-tempered RWM sampler. For each $d$, we consider 5 values of $\beta$, with the middle one being  the critical value $(d+3)/(d+5)$. 
}  \label{fig:violin}
\end{figure}

\end{document}